\documentclass[11pt]{article}
\usepackage{latexsym}

\usepackage[top=1in, bottom=1in, left=1in, right=1in]{geometry}

\usepackage{adjustbox}

\usepackage{amssymb}

\usepackage{graphicx} 
\usepackage{color} 
\usepackage{amsmath, amsthm, amssymb}
\usepackage{enumerate}
\usepackage{float}
\usepackage{url}
\usepackage{stackrel}
\usepackage{mathrsfs,dsfont}
\usepackage{caption}
\usepackage{subcaption}
\usepackage{wrapfig}
\usepackage{bbm}
\usepackage{bm}
\usepackage{epigraph}
\usepackage{physics}
\usepackage[colorinlistoftodos, shadow]{todonotes} 
\usepackage{multirow}

\usepackage{mdwlist}

\usepackage{pifont}

\newtheorem{theorem}{Theorem}[section]
\newtheorem{corollary}[theorem]{Corollary}
\newtheorem{proposition}[theorem]{Proposition}
\newtheorem{lemma}[theorem]{Lemma}
\newtheorem{example}[theorem]{Example}
\newtheorem{remark}[theorem]{Remark}

\theoremstyle{definition}
\newtheorem{definition}[theorem]{Definition}
\theoremstyle{definition}

\newcommand{\been}{\begin{enumerate}}
\newcommand{\enen}{\end{enumerate}}
\newcommand{\beit}{\begin{itemize}}
\newcommand{\enit}{\end{itemize}}

\def\PP{\mathcal P}

\def\GG{\mathcal{G}}

\def\XX{\mathcal{X}}

\def\one{{\mathds 1}}

\def\vp{\varphi}

\def\ds{\displaystyle}

\def\cre{\color{red}}

\newcommand{\R}{\mathbb{R}}

\newcommand{\Z}{\mathbb{Z}}

\usepackage{tikz}    
\usetikzlibrary{shapes,automata,positioning,arrows,fit}
\usepackage{mathtools}  
\usetikzlibrary{snakes}

\usepackage{caption}
\usepackage{tikz-cd}
\usetikzlibrary{babel}

\usepackage{tkz-euclide}

\newcommand{\specialcell}[2][c]{\begin{tabular}[#1]{@{}c@{}}#2\end{tabular}}

\numberwithin{equation}{section}

\usepackage{tikz-3dplot}

\usepackage{relsize}

 \tikzset{every node/.style={auto}}
 \tikzset{every state/.style={rectangle, minimum size=0pt, draw=none, font=\normalsize}}
  \tikzset{bend angle=7}
  
  \usepackage{array}

\usepackage{authblk}

\makeatletter
\newcommand{\xrightleftarrows}[2]{%
  \mathrel{\mathop{%
    \vcenter{\offinterlineskip\m@th
      \ialign{\hfil##\hfil\cr
        \hphantom{$\scriptstyle\mspace{8mu}{#1}\mspace{8mu}$}\cr
        \rightarrowfill\cr
        \vrule height0pt width 2em\cr
        \leftarrowfill\cr
        \hphantom{$\scriptstyle\mspace{8mu}{#2}\mspace{8mu}$}\cr
        \noalign{\kern-0.3ex}
      }%
    }%
  }\limits^{#1}_{#2}}%
}
\makeatother

\begin{document}



\title{Bifunctional enzyme provides absolute concentration robustness in \\ multisite covalent modification networks}

\author{Badal Joshi \footnote{Department of Mathematics, California State University San Marcos (bjoshi$@$csusm.edu)} ~ and Tung D. Nguyen \footnote{Department of Mathematics, Texas A \& M University (daotung.nguyen$@$tamu.edu)}}

\date{}

\maketitle

\begin{abstract}
\noindent Biochemical covalent modification networks exhibit a remarkable suite of steady state and dynamical properties such as multistationarity, oscillations, ultrasensitivity and absolute concentration robustness. 
This paper focuses on conditions required for a network  of this type to have a species with absolute concentration robustness. We find that the robustness in a substrate is endowed by its interaction with a bifunctional enzyme, which is an enzyme that has different roles when isolated versus when bound as a substrate-enzyme complex. When isolated, the bifunctional enzyme promotes production of more molecules of the robust species while when bound, the same enzyme facilitates degradation of the robust species. These dual actions produce robustness in the large class of covalent modification networks.
For each network of this type, we find the network conditions for the presence of robustness, the species that has robustness, and its robustness value. 
The unified approach of simultaneously analyzing a large class of networks for a single property, i.e. absolute concentration robustness, reveals the underlying mechanism of the action of bifunctional enzyme while simultaneously providing a precise mathematical description of bifunctionality. 

~

\end{abstract}

\section{Introduction}

The steady state properties and dynamics of covalent modification networks have been extensively studied since the work of Goldbeter and Koshland \cite{goldbeter1981amplified, straube2013reciprocal, jeynes2021ultrasensitivity}. A large class of such networks are multisite phosphorylation-dephosphorylation networks which play a key role in cell signaling circuits \cite{ThomGuna09a,walsh2006posttranslational,cohen2001role}. 
The number of phosphorylation sites can be as many as 150, resulting in as many as $2^{150}$ distinct phospho-forms \cite{gnad2007phosida}. 
The high dimensionalities of the state space and parameter space of the resulting dynamical model put it far beyond the scope of a detailed simulation study. 
The topology of the underlying reaction network helps guide the unraveling of complex dynamics \cite{ThomGuna09a,ThomGuna09,WangSontag,conradi2015global,holmberg2002multisite,manrai2008geometry,rubinstein2016long}. 
A parameterization of the positive steady states in the distributive reversible covalent modification  cycle (futile cycle) of arbitrary length  \cite{ThomGuna09} decreases the complexity from a larger number of differential equations to a small number of algebraic equations. 
The distributive futile cycle was shown \cite{WangSontag} to have between $n$ and $2n-1$ steady states where $n$ is the number of phosphorylation sites.

In this paper, we study the conditions on network topology that endow concentration robustness in covalent modification networks. 
Robustness in signal transduction against high variability in circuit components has been experimentally observed in thousands of bacterial signaling systems \cite{shinar2007input,alon2019introduction,hart2013utility}. 
Experimental and mathematical modeling work identified underlying the mechanism was a {\em bifunctional enzyme} (or {\em paradoxical enzyme}), a protein that performs opposing kinase and phosphatase activities \cite{russo1993essential,hsing1998mutations,batchelor2003robustness,shinar2009robustness,dexter2013dimerization}. 
In general, a bifunctional enzyme has distinct and possibly opposing effects on the network output. 

The specific notion of robustness that we consider in this paper is {\em absolute concentration robustness}, which refers to the property that the steady state concentration of a specific substrate remains invariant across all positive steady states, even when the initial values of all variables are allowed to vary.  According to the Shinar-Feinberg criterion \cite{shinar2010structural}, two simple network conditions suffice to ensure absolute concentration robustness: (i) the deficiency of the reaction network is one, and (ii) two non-terminal complexes differ in exactly one species (see Definitions \ref{def:def} and \ref{def:nonterminal} and Theorem \ref{thm:shinfein}). If these conditions are satisfied then the network has absolute concentration robustness in the species that is the difference of the two non-terminal complexes appearing in the second condition. 

The property of deficiency equal to one is useful for parameterization of positive steady states which helps prove the property of absolute concentration robustness in such cases. It is not clear that deficiency one plays a {\em functional} role in network output or dynamics. Therefore, one might expect that absolute concentration robustness is found even for reaction networks with deficiency not equal to one. In fact, zero deficiency networks can have absolute concentration robustness, the conditions for which are studied in \cite{joshi2022foundations}.  Another approach to establish ACR in networks without deficiency one involves network translation techniques \cite{tonello2018network}.  The networks studied in this paper can have arbitrary deficiency, and for every possible deficiency value, we find networks with the property of absolute concentration robustness. 

The second property that two non-terminal complexes differ in exactly one species {\em is} related to a functional property: {\em activity of a bifunctional enzyme}. In higher deficiency networks, bifunctionality needs a more careful definition since the same enzyme is often implicated in covalent modification of multiple sites. Thus in more complex networks, an enzyme may promote the increase of several different substrate types. Moreover, if the enzyme is bifunctional it may simultaneously facilitate the removal of many of these substrate types. 
In terms of the second Shinar-Feinberg condition, there can be several pairs of non-terminal complexes that differ in exactly one species and moreover, this differing species is usually different for every such pair. 

Our goal in this work is to find a {\em generalization of the second Shinar-Feinberg condition for a large class of covalent modification networks} while at the same time not relying on or requiring the first Shinar-Feinberg condition to hold. We address the questions of: (i) finding sufficient network conditions for absolute concentration robustness (ACR) to hold, (ii) finding species for which ACR holds (subject to the existence of a positive steady state), (iii) the ACR value of each ACR species, and (iv) necessary and sufficient conditions for the existence of a positive steady state.

The rest of the paper is organized as follows. Section \ref{sec:CRN} reviews basic definitions related to reaction networks and ACR. Section \ref{sec:model} introduces the detailed model for multisite covalent modification networks with a bifunctional enzyme, and provides preliminary analysis on their steady state equations. Section \ref{sec:ACR} establishes sufficient conditions for ACR in these networks. Section \ref{sec:existenceofss} gives necessary and sufficient conditions for the existence of a positive steady state in these networks. Section \ref{sec:parameterization} provides a steady state parameterization for futile cycles with a  bifunctional enzyme. Finally, we end with a discussion on boundary steady states of covalent modification networks with a bifunctional enzyme in Section \ref{sec:boundaryss}.

\section{Reaction network}\label{sec:CRN}

Here, we first recall the basic setup and definitions involving reaction networks, the dynamical systems they generate (Section 2.1), and absolute concentration robustness (Section 2.2). Readers who are familiar with reaction network theory can skip to Section \ref{sec:model} for the specific class of reaction networks we study in this paper. 

\subsection{Reaction networks and dynamical systems}
A {\em reaction network} $\GG$ is a directed graph whose vertices are non-negative linear combinations of {\em species} $X_1,X_2, \ldots, X_d$.  
In reaction network literature, we often refer to each vertex as a {\em complex}, and we denote a complex by $y= y_{1}X_1 + y_{2} X_2 + \dots + y_{d} X_d $  or by $y=( y_{1}, y_{2}, \dots, y_{d})$ 	(where $y_{i}\in\Z_{\geq 0}$).

Edges of $\GG$ represent the possible changes in the abundances of the species, and are referred to as {\em reactions}. 
The vector $y'-y$ is the {\em reaction vector} associated to the reaction $y \to y'$. Additionally, in this reaction,  $y$ is called the \textit{reactant complex}, and $y'$ is called the \textit{product complex}. If there is also a reaction from $y'$ to $y$, we write $y \rightleftarrows y'$ and say that they are a pair of {\em reversible reactions}.

\begin{example} \label{ex:futile}
An example of a reaction network $\GG$ is the cycle with two substrates $S_1,S_2$ whose interconversion is facilitated by enzymes $E$ and $F$:
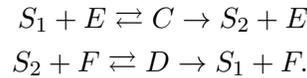
\begin{align*}
S_1+E \rightleftarrows C \rightarrow S_2+E \\
S_2+F \rightleftarrows D \rightarrow S_1+F.
\end{align*}
Here $\GG$ has 6 species $\{S_1,S_2,E,F,C,D\}$; 6 complexes $\{S_1+E,C,S_2+E,S_2+F,D,S_1+F\}$; and 6 reactions (one reaction for each arrow).
\end{example}
Under the assumption of mass-action kinetics, each reaction network $\GG$ defines a  parametrized family of systems of ordinary differential equations (ODEs), as follows. First, we enumerate the reactions by $y_i \to y_i'$ for $i=1,\dots,r$ and for each reaction $y_i\to y_i'$ we assign a positive {\em rate constant} $\kappa_{i} \in \mathbb{R}_{>0}$. Then the \textit{mass-action system}, denoted by $(\GG,\kappa)$ is the dynamical system arising from the following ODEs:
\begin{equation}\label{eq:mass_action_ODE}
		\frac{dx}{dt} ~=~ \sum_{i=1}^r  \kappa_{i} x^{y_i} (y_i'-y_i)~=:~ f_{\kappa}(x)~,
\end{equation}
where  $x(t)=(x_1(t),\dots,x_d(t))$ denotes the concentration of the species at time $t$ and $x^{y_i} := \prod_{j=1}^d x_j^{y_{ij}}$. The right-hand side of the ODEs~\eqref{eq:mass_action_ODE} consists of polynomials 
$f_{\kappa,i}$, for $i=1,2,\dots,d$ (where $d$ is the number of species). For simplicity, we often write $f_i$ instead of $f_{\kappa,i}$. 

The ODEs \eqref{eq:mass_action_ODE} can also be written in matrix form
\begin{equation}
    \frac{dx}{dt} = N \cdot v(x),
\end{equation}
where $N$, {\em the stoichiometric matrix}, is the matrix whose columns are all reaction vectors of $\GG$ and  $v_i(x)=\kappa_i x^{y_i}$. A {\em conservation law matrix} of $\GG$, denoted by $W$, is any row-reduced matrix whose rows form a basis of $\text{im}(N)^\perp$. The {\em conservation laws} of $\GG$ are given by $W x = c$, where $c:=W x(0)$ is the {\em total-constant vector}.

We denote by $S:=\text{im}(N)$, the \textit{stoichiometric subspace} of $\GG$.  Observe that the vector field of the mass-action ODEs~\eqref{eq:mass_action_ODE} lies in $S$ (more precisely, the vector of ODE right-hand sides is always in $S$).  Hence,  a forward-time solution $\{x(t) \mid t \ge 0\}$ of~\eqref{eq:mass_action_ODE}, with initial condition $x(0)  \in \mathbb{R}_{>0}^d$, remains in the following \textit{stoichiometric compatibility class}~\cite{feinberg2019foundations}: $$P_{x(0)} ~:=~ (x(0)+S)\cap \mathbb{R}_{\geq 0}^d~.$$

A {\em steady state} of a mass-action system is a nonnegative concentration vector $x^*\in \R_{\geq 0}^d$ at which the right-hand side of the ODEs \eqref{eq:mass_action_ODE} vanishes: $f_{\kappa}(x^*)=0$. 

\begin{definition}
$\GG$ is a {\em consistent reaction network} if there exist $\beta_{y \to y'} >0$ such that $\ds \sum_{y \to y' \in \GG} \beta_{y \to y'} (y' - y) = 0$.
\end{definition}

\begin{theorem}
The following are equivalent
\been
\item $\GG$ is a consistent reaction network.
\item There is a choice of positive rate constants $\kappa$ such that the mass-action system $(\GG,\kappa)$ has a positive steady state.
\enen
\end{theorem}
\begin{proof}
Suppose that $\GG$ is consistent. Then there exist $\beta_{y \to y'} > 0$ such that $\sum_{y \to y' \in \GG} \beta_{y \to y'} (y' - y) = 0$. Choose $\kappa_{y \to y'} = \beta_{y \to y'}$ for all $y \to y' \in \GG$. Then from \eqref{eq:mass_action_ODE}, $(1,1,\ldots, 1)$ is a positive steady state. 
Conversely suppose that $x^*$ is a positive steady state  for some choice of rate constants $\{\kappa_{y\to y'}\}$. Then consistency of $\GG$ follows from letting $\beta_{y \to y'} = \kappa_{y \to y'} x^{*y}$. 
\end{proof}

\subsection{Absolute concentration robustness}

In the context of reaction networks, absolute concentration robustness (ACR) can be formally defined at the level of systems and also networks, the latter is a significantly stronger property. 

\begin{definition}[ACR] \label{def:acr}
Let $X_i$ be a species of a reaction network $\GG$ with $r$ reactions.
\begin{enumerate}
\item For a fixed vector of positive rate constants $\kappa \in \mathbb{R}^r_{>0}$, 
	the mass-action system $(\GG,\kappa)$ has {\em absolute concentration robustness} (ACR) in $X_i$ if $(\GG,\kappa)$ has a positive steady state and in every positive steady state $x \in \R_{> 0}^n$ of the system, the value of $x_i$ (the concentration of $X_i$) is the same constant $x_i^*$. This value $x_i^*$ is the \textit{ACR value} of $X_i$. 
\item The reaction network $\GG$ has {\em absolute concentration robustness}  in species $X_i$ if $\GG$ is consistent and  furthermore, for any $\kappa'$ such that the mass-action system $(\GG,\kappa' > 0)$ has a positive steady state, $(\GG,\kappa')$ must have ACR in $X_i$. 
\end{enumerate}
\end{definition}

Note that $x_i^*$ is independent of the positive steady state of $(\GG,\kappa)$, but this value does depend on $\kappa$ in general, as in the next example.
A natural interpretation of ACR is that a particular steady state coordinate (corresponding the ACR species $X_i$) is independent of initial concentrations.

To show ACR in a species $X_i$ (either for a mass action system or a reaction network), one must show two things: the first is that a positive steady state exists and the second is that the concentration of $X_i$ is invariant across all positive steady states. 
In this paper, we first show the second property under the assumption that a positive steady state exists and later in Section \ref{sec:existenceofss}, we go on to show that a positive steady state exists for every case found to have the invariance property in some species. 

\begin{remark}
In this paper, the notion of ACR that we study has been referred to as  {\em static} ACR \cite{joshi2022foundations}, since it requires only knowledge of the positive steady states. A stronger notion called {\em dynamic} ACR \cite{joshi2022foundations} requires convergence of the ACR species concentration to the ACR value. Static ACR is necessary for dynamic ACR; in future work, we plan to study additional conditions required to ensure dynamic ACR. 
\end{remark}


\begin{example}  The following network
\[
A +B \xrightarrow{\kappa_1} 2B, ~ B \xrightarrow{\kappa_2}A
\]
is a classic example of a network with ACR \cite{shinar2010structural}.  Indeed, at all positive steady states, the concentration of species $A$ is $\kappa_2/\kappa_1$, and hence the network has ACR in $A$. 
\end{example}

Shinar and Feinberg \cite{shinar2010structural} proposed a network condition that guarantees ACR. We first provide the necessary terminology to state such a network condition.
\begin{definition} \label{def:def}
    The {\em deficiency} of a reaction network $\GG$ is given by
    \[
    \delta = C - \ell - s,
    \]
    where $C$ is the number of complexes, $\ell$ is the number of connected components, and $s$ is the dimension of the stoichiometric subspace of $\GG$.
\end{definition}
\begin{definition} \label{def:nonterminal}
    A {\em strong linkage class} of a reaction network $\GG$ is a maximal subset of its complexes that are strongly connected. A {\em terminal strong linkage class} is a strong linkage class in which there is no reaction from its complexes to complexes in another strong linkage class. Complexes not belonging to any terminal strong linkage class are called {\em non-terminal}.
\end{definition}
\begin{theorem}[Deficiency one condition \cite{shinar2010structural}] \label{thm:shinfein}
Let $\GG$ be a consistent reaction network with deficiency of one and such that it has two non-terminal complexes that differ only in species $S$. Then $\GG$ has ACR in $S$. 
\end{theorem}

While the network conditions in the theorem above can be checked easily, deficiency one is a particularly restrictive condition. Deficiency can easily increase with the addition, or discovery of new reactions (see \cite[Lemma 2.1]{anderson2022prevalence}) and thus it is unrealistic to expect biochemical reaction networks to have deficiency of exactly one. The deficiency one condition is neither sufficient nor necessary for ACR, and the reaction networks we study in the next section do not normally satisfy this condition (see Remark \ref{rem:deficiency}). 

\section{Multisite covalent modification network with a bifunctional enzyme}\label{sec:model}

\subsection{Model}

In this section, we give a detailed description of a {\em multisite covalent modification network} which uses a {\em bifunctional} (or {\em paradoxical}) {\em enzyme}. 

\begin{definition}[Multisite covalent modification network with a bifunctional enzyme]\label{def:model}
Let $(\GG,(k,h))$ denote the following mass-action system

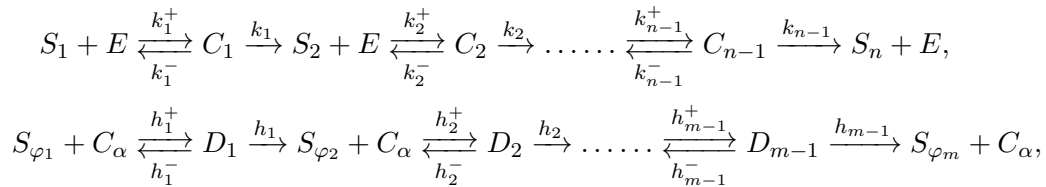
\begin{align*}
S_1 + E &\xrightleftarrows{k_1^+}{k_1^-} C_1 \xrightarrow{ k_1} S_2 + E \xrightleftarrows{k_2^+}{k_2^-} C_2 \xrightarrow{ k_2} \ldots \ldots \xrightleftarrows{k_{n-1}^+}{k_{n-1}^-} C_{n-1} \xrightarrow{ k_{n-1}} S_n + E, \\ 
S_{\vp_1} + C_\alpha &\xrightleftarrows{h_1^+}{h_1^-} D_1 \xrightarrow{ h_1} S_{\vp_2} + C_\alpha \xrightleftarrows{h_2^+}{h_2^-} D_2 \xrightarrow{ h_2} \ldots \ldots \xrightleftarrows{h_{m-1}^+}{h_{m-1}^-} D_{m-1} \xrightarrow{ h_{m-1}} S_{\vp_m} + C_\alpha, 
\end{align*}    
\beit
\item The top linkage class which converts $S_1$ to $S_n$,  through a sequence of steps, is referred to as the {\em forward chain}.  While the bottom linkage class which converts $S_{\vp_1}$ to $S_{\vp_m}$ is not necessarily the reverse of the forward chain, later we show in Section \ref{sec:existenceofss} that for $\GG$ to be consistent, the bottom linkage class must contain a path that transform $S_n$ into $S_1$. 
Since our main interest is in consistent networks, we refer to this linkage class as the {\em backward chain}.
\item $S_1, \ldots, S_N$ are {\em phosphoforms} or {\em substrates}. By convention, $\{S_1,\ldots, S_n\} \cup \{S_{\vp_1},\ldots, S_{\vp_m}\} \subseteq \{S_1,\ldots, S_N\}$ with the possibility of strict subset allowed for later convenience. 
\item $C_1, \ldots, C_{n-1}$ are the {\em intermediate complexes of the forward chain} while $D_1, \ldots, D_{m-1}$ are the {\em intermediate complexes of the backward chain}. 
\item $E$ is the {\em enzyme} for the forward chain while $C_\alpha$ is the {\em enzyme} for the backward chain. The index $\alpha$ takes values in $\{1,\ldots,n-1\}$ and therefore $C_\alpha$ is also one of the intermediate complexes of the forward chain. Because of the dual role of $C_\alpha$, it is referred to as a {\em bifunctional enzyme}. 
\item By assumption the indices $\vp_1, \ldots, \vp_m$ are distinct. The entire network $\GG$ is completely specified by the integers $n,m, \vp_1, \ldots, \vp_m$ and $\alpha$. Letting $\vp = (\vp_1,\ldots, \vp_m)$, we will refer to $\GG$ as {\em the covalent modification network $(n,m,\vp)$ with the bifunctional enzyme $C_\alpha$}. 
\item The reaction rate constants in the forward chain are labeled with a $k$ while in the backward chain, they are labeled with an $h$. Moreover, the complex formation and disassociation steps have the superscript $+$ and $-$ respectively while the subscript index indicates the covalent modification step/site. 
\enit
\end{definition}

The name {\em paradoxical} derives from the possibly contradictory roles of the complex $C_\alpha$, for example when the backward chain is the reverse of the forward chain. On the one hand, $C_\alpha$ promotes a step (eg. phosphorylation) in the forward chain, in the sense that an increase in the amount of $C_\alpha$ directly leads to increase in the form $S_{\alpha + 1}$. On the other hand, $C_\alpha$ acts as an enzyme in the backward chain (eg. promotes dephosphorylation).  

\begin{remark}
For $a < b \in \Z_{>0}$, we use the notation $[a,b]$ to mean the sequence $(a,a+1,\ldots,b)$. 
We can think of $\vp$ as an injective function mapping indices in $\{1,\dots,m\}$ to $\{1,\dots,N\}$, i.e. $\vp_j=\vp(j)$. We denote the image of this map by $\vp([1,m])$. We write $j = \vp^{-1}(i)$ if there exists an index $j\in \{1,\dots,m\}$ such that  $i=\vp_j=\vp(j)$.
\end{remark}

\begin{example}
An example of a covalent modification network with a bifunctional enzyme can be found in the E. coli IDHKP-IDH glyoxylate bypass regulation system \cite{shinar2010structural}:
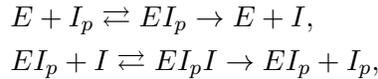
\begin{align*}
&E+I_p \rightleftarrows EI_p \rightarrow E+I, \\
&EI_p + I \rightleftarrows EI_pI \rightarrow EI_p+I_p,
\end{align*}
where $I$ denotes the active, unphosphorylated TCA cycle enzyme isocitrate dehydrogenase (IDH), and $I_p$ denotes the phosphorylated form. Here $EI_p$ (or $E$ in native form) is a bifunctional enzyme. Upto species relabelling, this network is equivalent to the covalent modification network $(2,2,\vp)$ where $\vp_1=2, \vp_2=1$ with the bifunctional enzyme $C_1$:
\begin{align*}
&S_1+E \rightleftarrows C_1 \rightarrow S_2+E\\
&S_2+C_1 \rightleftarrows D_1 \rightarrow S_1+C_1.
\end{align*}
This network is also the smallest among the networks in Definition \ref{def:model}.
\end{example}

\subsection{Biologically significant special cases}
Different choices of $n$, $m$, $\vp$ and $\alpha$ result in different networks. Certain specific choices have special relevance. 
\beit
\item {\bf covalent modification cycle (with bifunctional enzyme).} A covalent modification network where the initial and final substrates of the forward chain are the final and initial substrates of the backward chain is a {\em covalent modification cycle}. More concretely, a covalent modification cycle is a covalent modification network $(n,m,\vp)$ with $\vp_1 =n$ and $\vp_m = 1$. 
\item {\bf reversible covalent modification cycle = futile cycle (with bifunctional enzyme).} 
Every step in the forward chain is reversed in the backward chain. Concretely, this refers to a covalent modification network $(n,n,\vp)$ with $\vp([1,n]) = (n,n-1,\ldots, 2,1)$. 
\enit

\subsection{Steady state equations}
In this subsection, we provide some preliminary analysis on the deficiency and steady state equations of covalent modification networks with a bifunctional enzyme.
 The subsection is organized as follows. First, we state the ODEs of the mass-action system associated with $\GG$, and define some compact expressions that help with the analysis. In particular, Lemmas \ref{lem:cascade} and \ref{lem:cascade_ss} provide some useful relations between these expressions. Using these two lemmas, we provide an equivalent but simplified set of equations at the positive steady state in Proposition \ref{prop:FGHK}, Corollary \ref{cor:HK} and Corollary \ref{cor:steadystate}. Notably, Corollary \ref{cor:steadystate} is also used in obtaining a steady state parameterization in Section \ref{sec:parameterization}.

For convenience, we let ${k_{0}} = k_{0}^+=k_{0}^-={k_{n}} = k_{n}^+=k_{n}^-=0$ and ${h_{0}} = h_{0}^+=h_{0}^-={h_{m}} = h_{m}^+=h_{m}^-=0$  and $c_0=d_0=c_n=d_m=0$. The ODEs of the mass-action system $(G,(k,h))$ in Definition \ref{def:model} are given by
\begin{equation}\label{eq:ODEs}
\begin{aligned}
&\frac{ds_i}{dt} = \begin{cases} F_i+G^\alpha_{\vp^{-1}(i)}, \quad &\text{for}\quad i \in [1,n]  \cap \vp([1,m]),\\  F_i \quad &\text{for} \quad i\in [1,n]\cap  \vp([1,m])^c,  \\ G^\alpha_{\vp^{-1}(i)} \quad &\text{for} \quad i\in [1,n]^c\cap  \vp([1,m]), \end{cases} \\
&\frac{de}{dt}=\sum_{i=1}^n F_i,   \\
&\frac{dd_j}{dt}= - G^\alpha_j + H_j,\quad \text{for}\quad j\in[1,m-1],    \\
&\frac{dc_i}{dt}=-F_i+ K_i+\delta_\alpha(i)\sum_{j=1}^m G^\alpha_j, \quad \text{for}\quad i\in[1,n-1],
\end{aligned}
\end{equation}
where $\alpha$ is a fixed index in $[1,n-1]$ and 
\begin{equation}\label{eq:FGHK}
\begin{aligned}
&F_i := - k_i^+s_ie + k_i^-c_i + {k_{i-1}}c_{i-1}, \quad \text{for}\quad i \in [1,n], \\
&G_j^\alpha := -h_j^+s_{\vp_j}c_\alpha + h_j^-d_j +{h_{j-1}}d_{j-1}, \quad \text{for}\quad j \in [1,m], \\
&H_j :={h_{j-1}}d_{j-1}-{h_j}d_j, \quad \text{for}\quad j \in [1,m], \\
&K_i :={k_{i-1}}c_{i-1}-{k_i}c_i, \quad \text{for}\quad i \in [1,n]. 
\end{aligned}
\end{equation}
The following identities follow immediately:
\begin{align}\label{eq:KnHm}
F_n= K_n (=k_{n-1} c_{n-1}) \quad \text{and} \quad G^\alpha_m = H_m (=h_{m-1} d_{m-1}).
\end{align}
The system \eqref{eq:ODEs} satisfies the conservation laws
\begin{equation}\label{eq:conservation}
\begin{aligned}
&\sum_{i=1}^ns_i+\sum_{i=1}^{n-1}c_i+2\sum_{j=1}^{m-1}d_j = T_s, \\
&e+ \sum_{i=1}^{n-1}c_i+\sum_{j=1}^{m-1}d_j = T_e, 
\end{aligned}
\end{equation}
where $T_s$ and $T_e$ denote the total substrate and total enzyme, respectively.

\begin{theorem}\label{prop:deficiency}
Let $\GG$ be a covalent modification network $(n,m,\vp)$ with bifunctional enzyme $C_\alpha$. Then $\GG$ has deficiency
\[
\delta = 
\begin{cases} 
0, &\mbox{ if } [1,n] \cap \vp([1,m]) = \varnothing, \\
\# ( [1,n] \cap \vp([1,m]) ) -1, & \mbox{ otherwise, } 
\end{cases}
\]
where $\#( [1,n] \cap \vp([1,m]) )$ denotes the number of common indexes between $[1,n]$ and $\vp([1,m])$.
\end{theorem}
\begin{proof}
We observe that the dimensions of  the subspace spanned by the reaction vectors in the forward chain and the subspace spanned by the reaction vectors in the backward chain are $2n-2$ and $2m-2$ respectively. 
If $[1,n] \cap \vp([1,m]) = \varnothing$ then there are no linear dependence relations among the reaction vectors and so the dimension of the stoichiometric space is $2n - 2 + 2m-2$. 
Otherwise, if there are $\# ( [1,n] \cap \vp([1,m]) )- 1$ linear dependence relations, the dimension of the stoichiometric subspace of $\GG$ is
\[
s=2n-2+2m-2- (\# ( [1,n] \cap \vp([1,m]) )- 1) = 2n+2m-3 - \# ( [1,n] \cap \vp([1,m]) ).
\]
Thus $\GG$ has deficiency
\[
\delta =C-\ell-s= (2n+2m-2)-2-(2n+2m-3 - \# ( [1,n] \cap \vp([1,m]) ) ) =\# ( [1,n] \cap \vp([1,m]) ) -1, 
\]
and a similar calculation gives the stated deficiency for the case of $[1,n] \cap \vp([1,m]) = \varnothing$.
\end{proof}
\begin{remark}\label{rem:deficiency}
From Proposition \ref{prop:deficiency}, $\GG$ has deficiency one if and only if $\# ( [1,n] \cap \vp([1,m]) ) =2$, i.e. when the forward and backward chain have exactly two substrates in common. In all other cases, $\GG$ does not have deficiency one and thus the Shinar-Feinberg deficiency one theorem cannot be applied.
\end{remark}

Before examining the steady states of \eqref{eq:ODEs}, we state an important fact regarding the expressions $K$ and $H$.
\begin{lemma}[Cascading sums of $K$ and $H$]\label{lem:cascade}
The expressions $K$ and $H$ in \eqref{eq:FGHK} satisfy
\begin{enumerate}
\item[(1)] $\ds \sum_{\ell=1}^i K_\ell = -k_ic_i$ for $i=1,\dots, n$. In particular, $\ds \sum_{\ell=1}^n K_\ell =0.$
\item[(2)] $\ds \sum_{\ell=j+1}^m H_\ell= h_jd_j$ for $j=0,\ldots,m-1$. In particular, $\ds \sum_{\ell=1}^m H_\ell =0.$
\end{enumerate}
\end{lemma}
\begin{proof}
The equations can be easily derived from the definition of $K$ and $H$ in \eqref{eq:FGHK} and the fact that $k_0=k_n=0$ and $h_0=h_m=0$. 
\end{proof}

\begin{lemma}[Cascading sums of $G$ and $F$ at steady state]\label{lem:cascade_ss}
At steady state, we have the following: 
\been[(1)]
\item $\ds \sum_{i=1}^n F_i  = 0$.
\item $\ds \sum_{j=1}^m G^\alpha_j  = 0$.
\enen
\end{lemma}
\begin{proof}
The first claim is immediate from $de/dt$. To obtain the second equality, we first recall from \eqref{eq:KnHm}  that $G^\alpha_m = H_m$. Thus we have
\begin{align}\label{eq:sumG}
\sum_{j=1}^m G^\alpha_j = \sum_{j=1}^m H_j = 0,
\end{align}
where the second equality follows from Lemma \ref{lem:cascade}(2).
\end{proof}

\begin{proposition}\label{prop:FGHK}
The steady states  for the ODEs \eqref{eq:ODEs} are solutions of the following equations

\begin{enumerate}
\item[(1)] $F_i= -G^\alpha_{\vp^{-1}(i)} \quad \text{for}\quad i \in [1,n] \cap \vp([1,m])$,
\item[(2)] $F_i= 0 \quad \text{for}\quad i \in [1,n] \cap \vp([1,m])^c$,
\item[(3)] $G^\alpha_j=0 \quad \text{for} \quad \vp_j \in [1,n]^c \cap \vp([1,m])$,
\item[(4)] $G^\alpha_j = H_j \quad \text{for}\quad j\in[1,m-1]$,
\item[(5)] $F_i = K_i, \quad \text{for}\quad i\in[1,n-1]$. 
\end{enumerate}

\end{proposition}
\begin{proof}
Equations $(1)-(4)$ follow directly by setting $ds_i/dt=0$, $de/dt=0$ and $dd_j/dt=0$ in \eqref{eq:ODEs}. 
Setting $dc_i/dt=0$ in \eqref{eq:ODEs} and Lemma \ref{lem:cascade_ss}  yields $F_i = K_i$  for $i\in[1,n-1]$. 
\end{proof}

The next corollary, which will be used extensively in Section \ref{sec:ACR}, follows directly from Proposition \ref{prop:FGHK}.
\begin{corollary}\label{cor:HK}
The steady states  for the ODEs \eqref{eq:ODEs} satisfy the following equations
\begin{enumerate}
\item[(1)] $K_i = - H_{\vp^{-1}(i)}  \quad \text{for}\quad i \in [1,n] \cap \vp([1,m])$, 
\item[(2)] $K_i = 0  \quad \text{for}\quad i \in [1,n] \cap \vp([1,m])^c$, 
\item[(3)] $H_j=0 \quad \text{for} \quad \vp_j \in [1,n]^c \cap \vp([1,m])$. 
\end{enumerate}
\end{corollary}
\begin{proof}
We first note that $G^\alpha_m = H_m = h_{m-1}d_{m-1}$ and $F_n = K_n = k_{n-1}c_{n-1}$. Thus due to Proposition \ref{prop:FGHK}, the steady states for the ODEs \eqref{eq:ODEs} satisfy
\begin{itemize}
\item[($4'$)] $G^\alpha_j = H_j \quad \text{for}\quad j\in[1,m]$,
\item[($5'$)] $F_i = K_i, \quad \text{for}\quad i\in[1,n]$. 
\end{itemize}
Combining equation $(1)$ in Proposition \ref{prop:FGHK} and $(4'), (5')$ yields  $K_i = - H_{\vp^{-1}(i)}$ for $i \in [1,n] \cap \vp([1,m])$. Combining equation $(2)$  in Proposition \ref{prop:FGHK} and $(5')$ yields $K_i = 0$  for $i \in [1,n] \cap \vp([1,m])^c$. Finally, combining equation $(3)$  in Proposition \ref{prop:FGHK} and $(4')$ yields $H_j=0$ for $\vp_j \in [1,n]^c \cap \vp([1,m])$. 
\end{proof}

Next, we introduce two sets of quantities that are important in obtaining the ACR values in Section \ref{sec:ACR} and the steady state parameterization in Section \ref{sec:parameterization}.

\begin{definition} \label{def:star}
We define 
\[
k_i^* \coloneqq \frac{ k_i + k_i^-}{ k_i  k_i^+} \quad \text{for} \quad i\in [1, n-1]
\]
and
\[
h_j^* \coloneqq \frac{ h_j + h_j^-}{ h_j  h_j^+} \quad \text{for} \quad j\in [1, m-1].
\]
\end{definition}

\begin{corollary}\label{cor:steadystate}
The steady states for the ODEs \eqref{eq:ODEs} are solutions of the following equations
\begin{enumerate}

\item[(1)] $k_{i-1} c_{i-1} - k_{i} c_{i} = -h_{\vp^{-1}(i)-1} d_{\vp^{-1}(i)-1} + h_{\vp^{-1}(i)} d_{\vp^{-1}(i)}, \quad \text{for}\quad i \in [1,n] \cap \vp([1,m])$, 
\item[(2)] $k_{i-1}c_{i-1}=k_ic_i, \quad \text{for}\quad i  \in  [1,n] \cap \vp([1,m])^c$,  
\item[(3)] $h_{j-1}d_{j-1}=h_jd_j, \quad \text{for}\quad i = \vp_j \in  [1,n]^c \cap \vp([1,m])$,  
\item[(4)] $s_{\vp_j}c_\alpha  = h_j h_j^* d_j,\quad \text{for}\quad j\in[1,m-1]$, 
\item[(5)] $s_i e = k_i k_i^* c_i, \quad \text{for}\quad i\in[1,n-1]$. 
\end{enumerate}
\end{corollary}
\begin{proof}
Equations $(1), (4), (5)$  are equivalent to  equations $(1), (4), (5)$ in Proposition \ref{prop:FGHK}. Equation $(2)$ follows from combining equations $(2)$ and $(5)$ in Proposition \ref{prop:FGHK} and equation $(3)$ follows from combining equations $(3)$ and $(4)$ in Proposition \ref{prop:FGHK}.  Conversely, equations $(2)$ and $(5)$ here imply equation $(2)$ in Proposition \ref{prop:FGHK}, and equations $(3)$ and $(4)$ here imply equation $(3)$ in Proposition \ref{prop:FGHK}.
\end{proof}

\begin{remark}
Equation $(4)$  in Corollary \ref{cor:steadystate} is particularly important in identifying ACR species in Section \ref{sec:ACR}. It implies that if there exists an index $j$ such that the ratio $d_j/c_\alpha$ is a constant independent of the conserved quantities, then the network has ACR in species $\vp_j$. Consequently, the main approach in  Section \ref{sec:ACR} involves finding such an index $j$.
\end{remark}

\section{Bifunctional enzyme generates concentration robustness}\label{sec:ACR}
%
%

\begin{table}[htbp] 
\centering
\begin{adjustbox}{max width=\textwidth}
\renewcommand{\arraystretch}{1.2}
\begin{tabular}{|c|c|c||c|c|c|}
\hline
\specialcell{Type of\\operation} & \specialcell{Forward chain\\Index $(1, \ldots, n)$\\of $(S_1, \ldots, S_n)$} & \specialcell{Backward chain\\Index $(\vp_1, \ldots, \vp_m)$\\of $(S_{\vp_1}, \ldots, S_{\vp_m})$} & \specialcell{Enzyme\\complex\\$C_\alpha, \alpha =$} & \specialcell{ACR\\species\\$S_i, i=$} &\specialcell{ACR\\value} \\
\hline
\hline
\multirow{2}{*}{\specialcell{Futile\\cycle}} & \multirow{2}{*}{$(1,2,3)$} & \multirow{2}{*}{$(3,2,1)$} 
& $1$ & $2$ & $ k_1 h_2^*$\\
\cline{4-6}
& & & $2$ & $3$ & $ k_2 h_1^*$ \\
 \hline
\multirow{2}{*}{\specialcell{Deletion}} & \multirow{2}{*}{$(1,2,3)$} & \multirow{2}{*}{$(3,1)$} 
& 1 & 3 & $ k_1 h_1^*$\\
\cline{4-6}
& 
&  & 2 & 3 & $ k_2 h_1^*$\\
 \hline
 \multirow{2}{*}{\specialcell{Insertion}} & \multirow{2}{*}{$(1,2,3)$} & \multirow{2}{*}{$(3,4,2,1)$} 
& 1 & 2 & $ k_1 h_3^*$\\
\cline{4-6}
& 
&  & 2 & 3,4 & $ k_2 h_1^*$, $ k_2 h_2^*$\\
\hline 
\multirow{2}{*}{\specialcell{Deletion  \\ \& Insertion}} & \multirow{2}{*}{$(1,2,3)$} & \multirow{2}{*}{$(3,4,1)$} 
& 1 & 3,4 & $k_1 h_1^*$, $k_1 h_2^*$ \\
\cline{4-6}
& 
&  & 2 & 3,4 & $k_2 h_1^*$, $k_2 h_2^*$ \\
\hline 
\multirow{2}{*}{\specialcell{Permutation}} & \multirow{2}{*}{$(1,2,3)$} & \multirow{2}{*}{$(3,1,2)$} 
& 1 & -- & --\\
\cline{4-6}
& 
 & & 2 & 3 & $ k_2 h_1^*$\\
 \hline
\multirow{2}{*}{\specialcell{Permutation}} & \multirow{2}{*}{$(1,2,3)$} & \multirow{2}{*}{$(2,3,1)$} 
& 1 & 3 & $ k_1 h_2^*$\\
\cline{4-6}
& 
 & & 2 & -- & -- \\
 \hline
\multirow{2}{*}{\specialcell{Deletion}} &  \multirow{2}{*}{$(1,2,3,4)$} & \multirow{2}{*}{$(4,2,1)$} 
& 1 & 2 & $ k_1 h_2^*$\\
\cline{4-6}
& 
&  & 2 & 4 & $ k_2 h_1^*$\\
\cline{4-6}
& 
&  & 3 & 4 & $ k_3 h_1^*$\\
 \hline
 \multirow{2}{*}{\specialcell{Insertion}} & \multirow{2}{*}{$(1,2,3,4)$} & \multirow{2}{*}{$(4,3,5,2,1)$} 
& 1 & 2 & $ k_1 h_4^*$\\
\cline{4-6}
& 
&  & 2 & 3,5 & $ k_2 h_2^*$, $ k_2 h_3^*$\\
\cline{4-6}
& 
&  & 3 & 4 & $ k_3 h_1^*$\\
\hline 
\multirow{2}{*}{\specialcell{Permutation}} & \multirow{2}{*}{$(1,2,3,4)$} & \multirow{2}{*}{$(4,2,3,1)$} 
& 1 & 3 & $ k_1 h_3^*$\\
\cline{4-6}
& 
&  & 2 & -- & -- \\
\cline{4-6}
& 
&  & 3 & 4 & $ k_3 h_1^*$ \\
 \hline
\multirow{2}{*}{\specialcell{Permutation}} & \multirow{2}{*}{$(1,2,3,4)$} & \multirow{2}{*}{$(2,4,3,1)$} 
& 1 & 3 & $ k_1 h_3^*$\\
\cline{4-6}
& 
&  & 2 & -- & -- \\
\cline{4-6}
& 
&  & 3 & -- & -- \\
 \hline
\multirow{2}{*}{\specialcell{Permutation \\\& Insertion}} & \multirow{2}{*}{$(1,2,3,4)$} & \multirow{2}{*}{$(4,2,5,3,1)$} 
& 1 & 3 & $ k_1 h_4^*$\\
\cline{4-6}
& 
&  & 2 & -- & -- \\
\cline{4-6}
& 
&  & 3 & 4 & $ k_3 h_1^*$\\
\hline 
\multirow{2}{*}{\specialcell{Permutation \\\& Insertion \\\& Deletion}} & \multirow{2}{*}{$(1,2,3,4)$} & \multirow{2}{*}{$(4,5,1,3)$} 
& 1 & -- &  --\\
\cline{4-6}
& 
& & 2 & -- & --  \\
\cline{4-6}
& 
&  & 3 & 4,5 & $k_3h^*_1, k_3h^*_2$\\
\hline 
\end{tabular}
\end{adjustbox}
 \caption{{\bf Some examples of covalent modification networks and their ACR properties.} Every covalent modification network is obtained from some futile cycle by performing a series of insertion, deletion or permutation steps. For each network, the table indicates the ACR species and its ACR value given the choice of bifunctional enzyme $C_\alpha$.} \label{table:alphavsacr}
 \end{table}

In this section, we will show that ACR is a fairly generic property for covalent modification networks with a bifunctional enzyme. 
In particular, we will first show the existence of an ACR species in the futile cycle with bifunctional enzyme.

\begin{theorem}[futile cycle]\label{thm:base}
Let $\GG$ be a covalent modification network $(n,m,\vp)$ with a bifunctional enzyme $C_\alpha$, where  $m=n$ and $\vp_j=n+1-j$. 
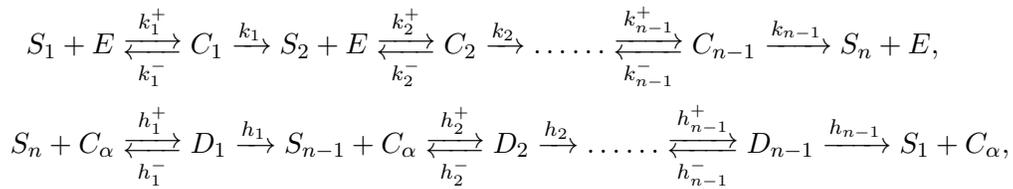
\begin{align*}
S_1 + E &\xrightleftarrows{k_1^+}{k_1^-} C_1 \xrightarrow{ k_1} S_2 + E \xrightleftarrows{k_2^+}{k_2^-} C_2 \xrightarrow{ k_2} \ldots \ldots \xrightleftarrows{k_{n-1}^+}{k_{n-1}^-} C_{n-1} \xrightarrow{ k_{n-1}} S_n + E, \\ 
S_{n} + C_\alpha &\xrightleftarrows{h_1^+}{h_1^-} D_1 \xrightarrow{ h_1} S_{n-1} + C_\alpha \xrightleftarrows{h_2^+}{h_2^-} D_2 \xrightarrow{h_2} \ldots \ldots \xrightleftarrows{h_{n-1}^+}{h_{n-1}^-} D_{n-1} \xrightarrow{ h_{n-1}} S_{1} + C_\alpha, 
\end{align*}    
Then $\GG$ has ACR in species $S_{\alpha+1}$ with ACR value $k_\alpha h^*_{n-\alpha}.$
\end{theorem}
\begin{proof}
We first observe from the assumption that $\vp_j=n+1-j$ that $\vp([1,m]) = [1,n]$ and $\vp^{-1}(i)=n+1-i$. Thus from Corollary \ref{cor:HK}, we have 
\[
K_i=-H_{n+1-i}\quad \text{for} \quad i\in[1,n].
\]
Thus we have
\[
\sum_{\ell=1}^i K_\ell =- \sum_{\ell=n+1-i}^n H_\ell  \quad \text{for} \quad i\in [1,n].
\]
Using Lemma \ref{lem:cascade}, we obtain 
\begin{equation}\label{eq:base_ratio}
k_i c_i = h_{n-i}d_{n-i} \quad \text{for} \quad i\in [1,n].
\end{equation}
In particular, setting $i=\alpha$ in the equation \eqref{eq:base_ratio} gives us
\[
k_\alpha c_\alpha = h_{n-\alpha}d_{n-\alpha}.
\]
Finally, setting $j=n-\alpha$ in equation $(4)$ of Corollary \ref{cor:steadystate} and observing that $\vp_{n-\alpha}=\alpha+1$, we obtain $s_{\alpha+1}=k_\alpha h^*_{n-\alpha}$.
\end{proof}

\begin{remark}
The steady state parameterization in Section \ref{sec:parameterization} further implies that for a futile cycle, $S_{\alpha+1}$ is the only ACR species.
\end{remark}

The most general network can be obtained from a futile cycle by doing a {\em finite sequence of only three operations on the backward chain}: (i) deletion, (ii) insertion, and (iii) permutation. 
For ease of reading, before providing the most general result, we will show in Theorems \ref{thm:deletion}, \ref{thm:insertion}, and \ref{thm:permutation} that the covalent modification networks obtained from performing each of the three operations on the futile cycle (permutations need to satisfy certain assumption) still have ACR.

\begin{theorem}[Deletion in the backward chain]\label{thm:deletion}
Let $\GG$ be a covalent modification network $(n,m,\vp)$ with a bifunctional enzyme $C_\alpha$, where $1=\vp_m < \vp_{m-1} < \dots < \vp_1=n$. Let the index $p$ be such that $\alpha\in [\vp_{p+1},\vp_{p}-1]$. Then $\GG$ has ACR in species $S_{\vp_{p}}$ with ACR value $k_\alpha h^*_p$.
\end{theorem}
\begin{proof}
We first observe from the assumption that $\vp_1,\dots,\vp_m \in [1,n]\cap \vp([1,m])$. Thus from equation $(1)$ in Corollary \ref{cor:HK}, we have
\begin{equation}\label{eq:del}
K_{\vp_{j}}=-H_{j} \quad \text{for} \quad j\in [1,m].
\end{equation}
Consequently, for each $j\in[1,m]$, we have
\[
\sum_{i=1}^{\vp_{j+1}} K_i = \sum_{\ell = j+1}^m K_{\vp_{\ell}} + \sum_{i\in[1,\vp_{j+1}]\cap\vp([1,m])^c} K_i = -\sum_{\ell=j+1}^m  H_\ell,
\]
where the second equality comes from \eqref{eq:del} and equation $(2)$ in Corollary \ref{cor:HK}. Using Lemma \ref{lem:cascade} we obtain
\begin{equation}\label{eq:del2}
k_{\vp_{j+1}}c_{\vp_{j+1}} = h_{j}d_{j} \quad \text{for} \quad j\in [1,m].
\end{equation}
Now if $\alpha = \vp_{p+1}$, then from \eqref{eq:del2} we have $ k_\alpha c_\alpha = h_{p}d_{p}$. If $\alpha>\vp_{p+1}$, then $[\vp_{p+1}+1,\alpha]\subset [1,n]\cap \vp([1,m])^c$. Thus using equation $(2)$ in Corollary \ref{cor:steadystate} yields
\[
k_\alpha c_\alpha = \dots = k_{\vp_{p+1}}c_{\vp_{p+1}}= h_{p}d_{p}.
\]
In both cases, we have $ k_\alpha c_\alpha = h_{p}d_{p}$. Finally, setting $j=p$ in equation $(4)$ in Corollary \ref{cor:steadystate} gives us $s_{\vp_p} =k_\alpha h^*_p$.
\end{proof}

The example below illustrates how Theorem \ref{thm:deletion} is used, and gives some intuition on the relation between the bifunctional enzyme and the substrate with ACR.
\begin{example}
Let $\GG$ be a covalent modification network $(4,3,\vp)$  with a bifunctional enzyme $C_\alpha$, where $\vp$ is given by $\vp_1=4, \vp_2=2$ and $\vp_3=1$. 
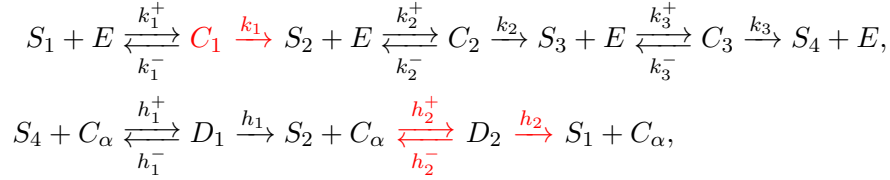
\begin{align*}
S_1 + E &\xrightleftarrows{k_1^+}{k_1^-} {\cre C_1 } ~ {\cre \xrightarrow{ k_1} } ~S_2 + E \xrightleftarrows{k_2^+}{k_2^-} C_2 \xrightarrow{ k_2} S_3+E \xrightleftarrows{k_3^+}{k_3^-} C_{3} \xrightarrow{ k_{3}} S_4 + E, \\ 
S_{4} + C_\alpha &\xrightleftarrows{h_1^+}{h_1^-} D_1 \xrightarrow{ h_1} S_{2} + C_\alpha ~ {\cre \xrightleftarrows{h_2^+}{h_2^-}} ~ D_2 ~ {\cre \xrightarrow{ h_2}} ~  S_{1} + C_\alpha, 
\end{align*}    
\begin{itemize}
\item When the bifunctional enzyme is $C_1$, we have $\alpha = 1 \in [\vp_3,\vp_2-1]$, thus $\GG$ has ACR in species $S_{\vp_2}=S_2$ with ACR value $k_1h^*_2$. The bifunctional enzyme $C_1$ and the reaction rate constants that appear in the ACR value are shown in red in the reaction network.  In this case, the ACR species $S_2$ can be seen as the target of the bifunctional enzyme $C_1$ as $C_1$ directly produces $S_2$, and catalyzes the degradation of $S_2$.
\item When the bifunctional enzyme is $C_2$ or $C_3$, we have $\alpha \in [\vp_2,\vp_1-1]$, thus in both cases $\GG$ has ACR in species $S_{\vp_1}=S_4$ with ACR values $k_2h^*_1$ and $k_3h^*_1$ respectively. When the bifunctional enzyme is $C_3$, again we can see $S_4$ as the target substrate as $C_3$ directly involves in the production and degradation of $S_4$. When the bifunctional enzyme is $C_2$, while $S_3$ is produced directly from $C_2$, it does not appear in the backward chain. So we have to look one step further to find the target substrate: $C_2$ indirectly involves in the production of $S_4$, and directly involves in the degradation of $S_4$. 
\end{itemize}
\end{example}

\begin{theorem}[Insertion in the backward chain]\label{thm:insertion}
Let $\GG$ be a covalent modification network $(n,m,\vp)$ with a bifunctional enzyme $C_\alpha$, where $m=\vp^{-1}(1) > \vp^{-1}(2) >\dots > \vp^{-1}(n)=1$. Then for any $j \in \{\vp^{-1}(\alpha+1),\dots, \vp^{-1}(\alpha) -1\}$, $\GG$ has ACR in species $\vp_j$ with ACR value $k_\alpha h^*_j$.
\end{theorem}
\begin{proof}
We first observe that $[1,n]\subset \vp([1,m])$. Thus from equation $(1)$ in Corollary \ref{cor:HK}, we have
\begin{equation}\label{eq:add}
K_{i}=-H_{\vp^{-1}(i)} \quad \text{for} \quad i\in [1,n].
\end{equation}
As a result, for any $i\in [1,n]$, we have
\begin{equation}\label{eq:add2}
\sum_{\ell =1}^i K_\ell = -\sum_{\ell=1}^i H_{\vp^{-1}(\ell)} = -\left(\sum_{j=\vp^{-1}(i)}^m H_{j} - \sum_{j\in[\vp^{-1}(i),m]:\vp_j \in [1,n]^c} H_j \right)= -\sum_{j=\vp^{-1}(i)}^m H_{j},
\end{equation}
where the third equality is due to equation $(3)$ in Corollary \ref{cor:HK}. In particular, setting $i=\alpha$ in \eqref{eq:add2} gives us
\[
\sum_{\ell =1}^\alpha K_\ell = -\sum_{j=\vp^{-1}(\alpha)}^m H_{j}.
\]
Using Lemma \ref{lem:cascade} we obtain 
\[
k_\alpha c_\alpha = h_{\vp^{-1}(\alpha)-1} d_{\vp^{-1}(\alpha)-1}.
\]
Furthermore, if $\vp^{-1}(\alpha+1) < \vp^{-1}(\alpha)- 1$, then $\vp([\vp^{-1}(\alpha+1)+1,\vp^{-1}(\alpha)- 1]) \in [1,n]^c$. Thus by equation $(3)$ in Corollary \ref{cor:steadystate} we have
\[
h_{\vp^{-1}(\alpha+1)}d_{\vp^{-1}(\alpha+1)}= \dots = h_{\vp^{-1}(\alpha)- 1)}d_{\vp^{-1}(\alpha)- 1} = k_\alpha c_\alpha.
\]
Finally, setting $j=\vp^{-1}(\alpha+1),\dots, \vp^{-1}(\alpha)-1$ in equation $(4)$ in Corollary \ref{cor:steadystate} yields 
\[
s_{\vp_j} = k_\alpha h^*_j \quad \text{for} \quad j=\vp^{-1}(\alpha+1),\dots, \vp^{-1}(\alpha)-1.
\]
\end{proof}

We illustrate the result of Theorem \ref{thm:insertion} in the following example, and again include some intuition on the relation of the bifunctional enzyme and the ACR substrate.

\begin{example}
Let $\GG$ be a covalent modification network $(3,4,\vp)$  with a bifunctional enzyme $C_\alpha$, where $\vp$ is given by $\vp_1=3, \vp_2=4, \vp_3=2$ and $\vp_4=1$. 
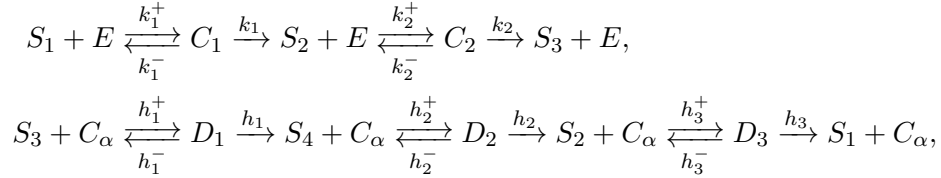
\begin{align*}
S_1 + E &\xrightleftarrows{k_1^+}{k_1^-} C_1 \xrightarrow{ k_1} S_2 + E \xrightleftarrows{k_2^+}{k_2^-} C_2 \xrightarrow{ k_2} S_3+E , \\ 
S_{3} + C_\alpha &\xrightleftarrows{h_1^+}{h_1^-} D_1  \xrightarrow{ h_1} S_{4} + C_\alpha \xrightleftarrows{h_2^+}{h_2^-} D_2 \xrightarrow{ h_2} S_{2} + C_\alpha \xrightleftarrows{h_3^+}{h_3^-} D_3 \xrightarrow{ h_3}  S_{1} + C_\alpha, 
\end{align*}    
\begin{itemize}
\item When the bifunctional enzyme is $C_1$ (i.e. $\alpha=1$), we have $\vp^{-1}(\alpha+1) = 3 = \vp^{-1}(\alpha)-1$, thus $\GG$ has ACR in species $S_{\vp_3}=S_2$ with ACR value $k_1h^*_3$.  In this case, the ACR species $S_2$ is the target of the bifunctional enzyme $C_1$ as $C_1$ directly produces $S_2$, and catalyzes the degradation of $S_2$.
\item When the bifunctional enzyme is $C_2$ (i.e. $\alpha=2$), we have $\vp^{-1}(\alpha+1) = 1$ and $\vp^{-1}(\alpha)-1 = 2$, thus $\GG$ has ACR in species $S_{\vp_1}=S_3$ and $S_{\vp_2}=S_4$ with ACR values $k_2h^*_1$ and $k_2h^*_2$ respectively. Similar to the previous case, here $S_3$ can be clearly identified as the target substrate as $C_2$ directl involves in its production and degradation. It is much more subtle to interpret $S_4$ as another target substrate. While $S_4$ does not appear in the forward chain, one could argue the closest enzyme from the forward chain responsible for its production is $C_2$ since $C_2$ produces $S_3$, which in turn transforms into $S_4$ in the backward chain. 
\end{itemize}
\end{example}

\begin{theorem}[Permutation in the backward chain]\label{thm:permutation}
Let $\GG$ be a covalent modification network $(n,m,\vp)$ with a bifunctional enzyme $C_\alpha$, where $\{\vp_m,\vp_{m-1},\dots,\vp_{m-\alpha+1}\}$ is a permutation of $\{1,\dots,\alpha\}$, then $\GG$ has ACR in species $S_{\vp_{m-\alpha}}$ with ACR value $k_\alpha h^*_{m-\alpha}$
\end{theorem}
\begin{proof} 
We first observe that $[1,\alpha]\subset [1,n]\cap \vp([1,m])$. Thus from equation $(1)$ in Corollary \ref{cor:HK}, we have
\[
K_i =- H_{\vp^{-1}(i)}\quad \text{for} \quad i\in [1,\alpha].
\]
Thus we have
\[
\sum_{i=1}^\alpha K_i = -\sum_{i=1}^\alpha  H_{\vp^{-1}(i)} = -\sum_{j=m-\alpha+1}^m H_j,
\]
where the second equality is due to the assumption that  $\{\vp_m,\vp_{m-1},\dots,\vp_{m-\alpha+1}\}$ is a permutation of $\{1,\dots,\alpha\}$. As a result, from Lemma \ref{lem:cascade} we obtain 
\[
k_\alpha c_\alpha = h_{m-\alpha}d_{m-\alpha}.
\]
Finally, setting $j=m-\alpha$ in equation $(4)$ in Corollary \ref{cor:steadystate}, we have $s_{\vp_{m-\alpha}}=k_\alpha h^*_{m-\alpha}$.
\end{proof}

We illustrate the result in Theorem \ref{thm:permutation} with an example below. The permutation operation is complicated and it is difficult to interpret the relation between the bifunctional enzyme and the ACR substrate. We plan to develop a much broader theoretical framework to explain this relation in a future paper.

\begin{example}
Let $\GG$ be a covalent modification network $(4,4,\vp)$  with a bifunctional enzyme $C_\alpha$, where $\vp$ is given by $\vp_1=4, \vp_2=2, \vp_3=3$ and $\vp_4=1$. 
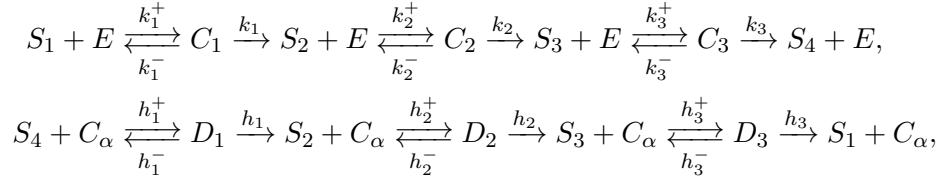
\begin{align*}
S_1 + E &\xrightleftarrows{k_1^+}{k_1^-} C_1 \xrightarrow{ k_1} S_2 + E \xrightleftarrows{k_2^+}{k_2^-} C_2 \xrightarrow{ k_2} S_3+E \xrightleftarrows{k_3^+}{k_3^-} C_{3} \xrightarrow{ k_{3}} S_4 + E, \\ 
S_{4} + C_\alpha &\xrightleftarrows{h_1^+}{h_1^-} D_1  \xrightarrow{ h_1} S_{2} + C_\alpha \xrightleftarrows{h_2^+}{h_2^-} D_2 \xrightarrow{ h_2} S_{3} + C_\alpha \xrightleftarrows{h_3^+}{h_3^-} D_3 \xrightarrow{ h_3}  S_{1} + C_\alpha, 
\end{align*}    
\begin{itemize}
\item When the bifunctional enzyme is $C_1$ (i.e. $\alpha=1$), we have $\vp_4=1=\alpha$, thus $\GG$ has ACR in species $S_{\vp_3}=S_3$ with ACR value $k_1h^*_3$. 
\item When the bifunctional enzyme is $C_2$ (i.e. $\alpha=2$), we have $\{\vp_4,\vp_3\} = \{1,3\} \neq \{1,2\}$. It can be checked that $\GG$ has no ACR species.
\item When the bifunctional enzyme is $C_3$ (i.e. $\alpha=3$), we have $\{\vp_4,\vp_3,\vp_2\} = \{1,2,3\}$, thus $\GG$ has ACR in species $S_{\vp_1}=S_4$ with ACR value $k_3h^*_1$.
\end{itemize}
\end{example}

\begin{theorem}[General network]\label{thm:general}

Let $\GG$ be a covalent modification network $(n,m,\vp)$ with a bifunctional enzyme $C_\alpha$.
Assume that $\{1,n\} \subseteq \vp([1,m])$. 
Let
\[ 
\beta = \min\{j: \vp_j\in [1,\alpha]\cap\vp([1,m])\}, \quad \text{and} \quad \gamma = \max\{j: \vp_j\in [\alpha+1, n]\cap\vp([1,m])\}.
\]
Then $\beta$ and $\gamma$ are defined.
Furthermore, if $\gamma<\beta$ then for any $j \in \{\gamma,\dots,\beta-1\}$, $\GG$ has ACR in species $S_{\vp_j}$ with ACR value $k_\alpha h^*_j$.
\end{theorem}

\begin{proof}
Since $1 \in \vp([1,m])$,  we must have $[1,\alpha]\cap\vp([1,m])\neq \emptyset$ and thus $\beta$ is defined. Similarly,  $[\alpha+1, n]\cap\vp([1,m])\neq \emptyset$, and thus $\gamma$ is also defined.
We observe that
\begin{equation}\label{eq:gen1}
\sum_{i=1}^\alpha K_i = \sum_{i\in[1,\alpha]\cap\vp([1,m])} K_i  + \sum_{i\in[1,\alpha]\cap\vp([1,m])^c} K_i = \sum_{i\in[1,\alpha]\cap\vp([1,m])} K_i,
\end{equation}
where the second equality is due to equation $(2)$ in Corollary \ref{cor:HK}. Next, we have
\begin{equation}\label{eq:gen2}
\sum_{j=\gamma+1}^{m} H_j = \sum_{j: \vp_j\in[1,n]\cap\vp([\gamma+1,m])} H_j +  \sum_{j: \vp_j\in[1,n]^c\cap\vp([\gamma+1,m])} H_j =  \sum_{j: \vp_j\in[1,n]\cap\vp([\gamma+1,m])} H_j,
\end{equation}
where the second equality is due to equation $(3)$ in Corollary \ref{cor:HK}. Before proceeding further, we prove the following claim.

\vspace{1mm}
\noindent \textbf{Claim:} $[1,n]\cap\vp([\gamma+1,m]) = [1,\alpha]\cap\vp([1,m])$.
\begin{itemize}
\item First, assume that $i\in [1,n]\cap\vp([\gamma+1,m])$. Then clearly we have $i\in\vp([1,m])$. Furthermore, since $i\in \vp([\gamma+1,m])$, we must have $\vp^{-1}(i) >\gamma$. From the definition of $\gamma$, this means $i\notin [\alpha+1,n]$. Thus $i\in [1,\alpha]$, which further implies $i\in [1,\alpha]\cap\vp([1,m])$.
\item For the remaining direction, assume that $i\in [1,\alpha]\cap\vp([1,m])$. Then clearly we have $i\in [1,n]$. From the definition of $\beta$ and the assumption that $\gamma<\beta$, we have $\vp^{-1}(i) \geq  \beta > \gamma$. Thus $i\in\vp([\gamma+1,m])$, which implies $i\in[1,n]\cap\vp([\gamma+1,m])$.
\end{itemize}
Using the above claim and equation $(1)$ in Corollary \ref{cor:HK}, we obtain 
\begin{equation}\label{eq:gen3}
\sum_{i\in[1,\alpha]\cap\vp([1,m])} K_i = \sum_{i\in[1,\alpha]\cap\vp([1,m])} H_{\vp^{-1}(i)} = \sum_{j: \vp_j\in[1,n]\cap\vp([\gamma+1,m])} H_j.
\end{equation}
Thus combining equations \eqref{eq:gen1}, \eqref{eq:gen2} and \eqref{eq:gen3} yields
\[
\sum_{i=1}^\alpha K_i  = \sum_{j=\gamma+1}^{m} H_j. 
\]
From Lemma \ref{lem:cascade} we obtain
\[
k_\alpha c_\alpha = h_{\gamma} d_{\gamma}. 
\]
Furthermore, if $\gamma < \beta-1$, then by the definitions of $\beta$ and $\gamma$ we have $\vp_j\in[1,n]^c\cap \vp([1,m])$ $\forall j\in [\gamma+1,\beta-1]$. Using equation $(3)$ in Corollary \ref{cor:steadystate}, we obtain
\[
  h_{\beta-1} d_{\beta-1}= \dots = h_{\gamma}d_{\gamma} = k_\alpha c_\alpha.
\]
Finally, setting $j=\gamma,\dots,\beta-1$ in equation $(4)$ in  Corollary \ref{cor:steadystate} yields
\[
s_{\vp_j} = k_\alpha h_j^* \quad \text{for} \quad j=\gamma,\dots,\beta-1.
\]
\end{proof}

\begin{remark}
We provide some intuition on the assumptions in Theorem \ref{thm:general}:
\begin{itemize}
\item The index $\beta$ indicates which species among $S_1,\dots,S_\alpha$ appears first in the backward chain.
\item The index $\gamma$ indicates which species among $S_{\alpha+1},\dots,S_n$ appears last in the backward chain.
\item The assumption $\gamma<\beta$ will always hold if we perform an ACR-preserving permutation in the backward chain (according to Theorem \ref{thm:permutation}) first, then finite number of insertions and deletions in the backward chain (as described in Theorems \ref{thm:deletion} and \ref{thm:insertion}). 
\end{itemize}
\end{remark}

\begin{example}
Let $\GG$ be a covalent modification network $(4,4,\vp)$  with a bifunctional enzyme $C_\alpha$, where $\vp$ is given by $\vp_1=4, \vp_2=5, \vp_3=1$ and $\vp_4=2$. 
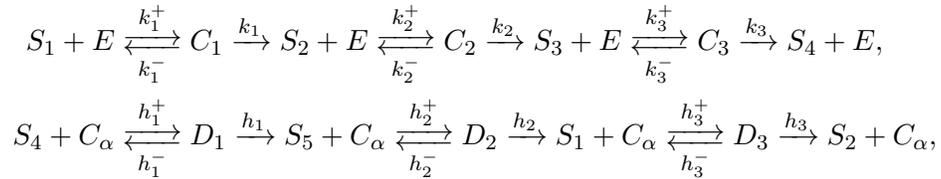
\begin{align*}
S_1 + E &\xrightleftarrows{k_1^+}{k_1^-} C_1 \xrightarrow{ k_1} S_2 + E \xrightleftarrows{k_2^+}{k_2^-} C_2 \xrightarrow{ k_2} S_3+E \xrightleftarrows{k_3^+}{k_3^-} C_{3} \xrightarrow{ k_{3}} S_4 + E, \\ 
S_{4} + C_\alpha &\xrightleftarrows{h_1^+}{h_1^-} D_1  \xrightarrow{ h_1} S_{5} + C_\alpha \xrightleftarrows{h_2^+}{h_2^-} D_2 \xrightarrow{ h_2} S_{1} + C_\alpha \xrightleftarrows{h_3^+}{h_3^-} D_3 \xrightarrow{ h_3}  S_{2} + C_\alpha, 
\end{align*}    
\begin{itemize}
\item When the bifunctional enzyme is $C_1$ (i.e. $\alpha=1$), we have $\beta = 3$ and $\gamma = 4$. Thus the assumption in Theorem \ref{thm:general} is not satisfied and it can be checked that there is no ACR species.
\item When  the bifunctional enzyme is $C_2$ (i.e. $\alpha=2$), we have $\beta = 3$ and $\gamma = 1$. Thus  the assumption in Theorem \ref{thm:general} is satisfied, and $\GG$ has ACR in species $S_{\vp_1}= S_4$ and $S_{\vp_2}= S_5$ with ACR values $k_2h^*_1$ and $k_2h^*_2$ respectively.
\item When  the bifunctional enzyme is $C_3$ (i.e. $\alpha=3$), we have $\beta = 3$ and $\gamma = 1$.  Again,  the assumption in Theorem \ref{thm:general} is satisfied, and $\GG$ has ACR in species $S_{\vp_1}= S_4$ and $S_{\vp_2}= S_5$ with ACR values $k_3h^*_1$ and $k_3h^*_2$ respectively. 
\end{itemize}
\end{example}

\begin{remark}
In Section \ref{sec:parameterization}, we prove that the futile cycle with bifunctional enzyme  has ACR in one and only one species. 
Insertion in the backward chain (and any combination of operations containing it) can give rise to more than one ACR species. As for deletion and ACR-preserving permutation in the backward chain, we believe there isn't ACR in any other species. Numerical simulations or parameterization of the type in  Section \ref{sec:parameterization} can help with ruling out ACR in other species not stated in Theorems \ref{thm:deletion} and \ref{thm:permutation}.
\end{remark}

We include some additional examples in Table \ref{table:alphavsacr}, where we apply the theorems in this section to find the ACR species in covalent modification networks with a bifunctional enzyme.


\begin{remark}\label{rem:multi_enzyme}
It is worth noting that our main result in Theorem \ref{thm:general} does not require $C_\alpha$ to be the only enzyme in the backward chain. In fact, the result still holds if $C_\alpha$ is replaced by another enzyme $F$  in any complex in the backward chain except for the complexes containing the ACR species. 
For example, consider the futile cycle with $n=3$ and $\alpha=1$, which has ACR in $S_2$:

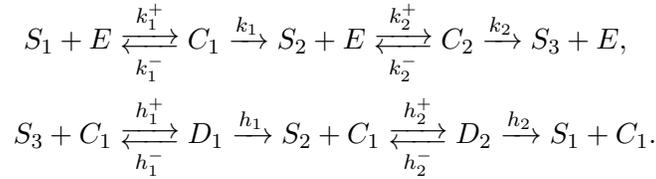
\begin{align*}
S_1 + E &\xrightleftarrows{k_1^+}{k_1^-} C_1 \xrightarrow{ k_1} S_2 + E \xrightleftarrows{k_2^+}{k_2^-} C_2 \xrightarrow{ k_2} S_3+E , \\ 
S_{3} + C_1 &\xrightleftarrows{h_1^+}{h_1^-} D_1  \xrightarrow{ h_1} S_{2} + C_1 \xrightleftarrows{h_2^+}{h_2^-} D_2 \xrightarrow{ h_2} S_{1} + C_1.
\end{align*}  
A variant of this network where some $C_1$ in the backward chain are replaced by $F$ still have ACR in species $S_2$:
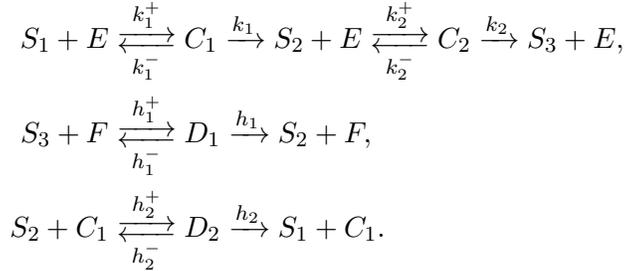
\begin{align*}
S_1 + E &\xrightleftarrows{k_1^+}{k_1^-} C_1 \xrightarrow{ k_1} S_2 + E \xrightleftarrows{k_2^+}{k_2^-} C_2 \xrightarrow{ k_2} S_3+E , \\ 
S_{3} + F &\xrightleftarrows{h_1^+}{h_1^-} D_1  \xrightarrow{ h_1} S_{2} + F, \\
S_2+C_1 &\xrightleftarrows{h_2^+}{h_2^-} D_2 \xrightarrow{h_2} S_{1} + C_1. 
\end{align*}
This type of replacements generally does not change our results. The main difference it brings forth lie in equation $(4)$ of Corollary \ref{cor:steadystate}, where $c_\alpha$ is replaced by $f$ for some $j$. Thus, the proof with this type of replacements remains mostly identical to the proof of Theorem \ref{thm:general} with very minor changes in notations. 

Of course, it is possible to encounter many other ACR-preserving variants of the class of networks studied in this paper. In future work, we will give results that significantly generalize the results in this paper.
\end{remark}

\section{Existence of positive steady state} \label{sec:existenceofss}

We give necessary and sufficient conditions for a covalent modification network to be consistent. 
Even though the theorem and proof are stated for a covalent modification network with a bifunctional enzyme, the result applies to any covalent modification network with or without a bifunctional enzyme with a minor modification.

\begin{theorem} \label{thm:consistency}
Let $\GG$ be a covalent modification network $(n,m,\vp)$. Define an auxiliary graph $N_\GG$ whose vertices are 
\[
\{X_1, \ldots, X_n\} \cup \{X_{\vp_1},\ldots, X_{\vp_m}\}
\]
and $X_\ell \to X_s$ is a directed edge of $N_\GG$ if and only if 
$\ell + 1 = s \le n$ or $\vp^{-1}(\ell) + 1 = \vp^{-1}(s) \le m$. 
The following are equivalent:
\been
\item $\GG$ is consistent, 
\item $N_\GG$ is strongly connected. 
\enen
\end{theorem}
\noindent {\bf Note:} The proof that ${\it 2.} \implies {\it 1.}$ is by direct construction. Specifically, we give rate constants of all reactions in $\GG$ and show that $(1,1,\ldots, 1)$ is a positive steady state for the chosen rate constants. 
The procedure to obtain the desired rate constants is explained in Remark \ref{rem:construction_existence} which follows the proof. 
\begin{proof}
Suppose that $N_\GG$ is strongly connected. 
Then, in particular, there is a path from $X_n$ to $X_1$. 
It follows that $n, 1 \in \Im (\vp)$ and $\vp^{-1}(n) < \vp^{-1}(1)$. 
Define rate constants in backward chain of $\GG$ as follows: 
\begin{equation} \label{eq:definebackrates}
\begin{aligned}
h_j = \begin{cases}
3 & \mbox{ if } \vp^{-1}(n) \le j < \vp^{-1} (1), \\
1 & \mbox{ otherwise}. 
\end{cases}
\end{aligned}
\end{equation}
Moreover, let $h_j^- = 1$ and $h_j^+ = h_j + 1$ for $j \in [1,m-1]$. 

Since $N_\GG$ is strongly connected, there must exist a node $X_\ell$ such that $X_\ell \to X_{\vp_1}$ . That means either $\ell+1 = \vp_1\leq n$ or $\vp^{-1}(\ell)+1 = \vp^{-1}(\vp_1)=1$. The latter is clearly impossible, and the former implies that $\vp_1 \in [2,n]$.
Similarly, strong connectedness of $N_\GG$ implies that $\vp_m \in [1,n-1]$. 
Define rate constants in forward chain of $\GG$ as follows: 
\begin{equation} \label{eq:defineforrates}
\begin{aligned}
k_i = \begin{cases}
1 & \mbox{ if } \vp_1 \le i < \vp_m, \\
3 & \mbox{ if } \vp_m \le i < \vp_1, \\
2 & \mbox{ otherwise}. 
\end{cases}
\end{aligned}
\end{equation}
Moreover, let $k_i^- = 1$ and $k_i^+ = k_i + 1$ for $i \in [1,n-1]$. 

To show that $\one \coloneqq (1,1,\ldots,1)$ is a steady state, it suffices to show that the production rate of every species is equal to its consumption rate at the state $\one$, i.e. for every species $s$ in $\GG$, the following must hold:
\begin{align} \label{eq:balancecondition}
\sum_{s \in y'} \kappa_{y \to y'} - \sum_{s \in y} \kappa_{y \to y'}   = 0
\end{align}
where $\kappa_{y \to y'}$ is the reaction rate constant of the reaction $y \to y'$ and $s \in y$ means that the species $s$ has positive stoichiometric coefficient in the complex $y$. 
\beit
\item For the species $S_1$, we check that 
\begin{align*}
\sum_{S_1 \in y'} \kappa_{y \to y'} - \sum_{S_1 \in y} \kappa_{y \to y'} &= \left[k_0 + k_1^- - k_1^+\right] + \left[h_{\vp^{-1}(1)-1} + h_{\vp^{-1}(1)}^- - h_{\vp^{-1}(1)}^+\right] \\ 
&=
\left.
\begin{cases}
 \left[0 + 1 - (1+2)\right] + \left[3 + 1 - (1+1)\right] &\mbox{ if }  \vp_m \ne 1 \\
 \left[0 + 1 - (1+3)\right] + \left[3 + 0 - 0\right] &\mbox{ if }  \vp_m = 1
\end{cases} \right\}= 0.
\end{align*}
\item For the species $S_n$, we check that 
\begin{align*}
\sum_{S_n \in y'} \kappa_{y \to y'} - \sum_{S_n \in y} \kappa_{y \to y'} &= \left[k_{n-1} + k_n^- - k_n^+\right] + \left[h_{\vp^{-1}(n)-1} + h_{\vp^{-1}(n)}^- - h_{\vp^{-1}(n)}^+\right] \\ 
&=
\left.
\begin{cases}
 \left[2 + 0 - 0 \right] + \left[1 + 1  - (1+3)\right] &\mbox{ if }  \vp_1 \ne n \\
 \left[3 + 0 - 0 \right] + \left[0 + 1 - (1+3)\right] &\mbox{ if }  \vp_1 = n
\end{cases} \right\}= 0.
\end{align*}
\item For the species $S_{\vp_1}$, assuming that $\vp_1 \ne n$ (since this case is already covered), we check that 
\begin{align*}
\sum_{S_{\vp_1} \in y'} \kappa_{y \to y'} - \sum_{S_{\vp_1} \in y} \kappa_{y \to y'} &= \left[k_{\vp_1-1} + k_{\vp_1}^- - k_{\vp_1}^+\right] + \left[h_{0} + h_{1}^- - h_{1}^+\right] \\ 
&=
\left.
\begin{cases}
\left[ 2 + 1 - (1+1)\right] + \left[ 0 + 1 - (1+1) \right] &\mbox{ if }  \vp_1 <  \vp_m \\
\left[ 3 + 1 - (1+2)\right] + \left[ 0 + 1 - (1+1) \right] &\mbox{ if }  \vp_1 >  \vp_m 
\end{cases}
\right\}= 0.
\end{align*}
\item For the species $S_{\vp_m}$, assuming that $\vp_m \ne 1$ (since this case is already covered), we check that 
\begin{align*}
\sum_{S_{\vp_m} \in y'} \kappa_{y \to y'} - \sum_{S_{\vp_m} \in y} \kappa_{y \to y'} &= \left[k_{\vp_m-1} + k_{\vp_m}^- - k_{\vp_m}^+\right] + \left[h_{m-1} + h_{m}^- - h_{m}^+\right] \\ 
&=
\left.
\begin{cases}
\left[ 1 + 1 - (1 + 2) \right] + \left[ 1 + 0 - 0 \right] &\mbox{ if }  \vp_1 <  \vp_m \\
\left[ 2 + 1 - (1+3)\right] + \left[ 1 + 0 - 0  \right] &\mbox{ if }  \vp_1 >  \vp_m 
\end{cases}
\right\}= 0.
\end{align*}
\item For the species $S_i$, $i \notin \{1,n,\vp_1,\vp_m\}$, it is easy to check that the rate constants balance ``locally'', i.e. $k_{i-1} + k_{i}^- - k_{i}^+ = 0$ and $h_{\vp^{-1}(i)-1} + h_{\vp^{-1}(i)}^- - h_{\vp^{-1}(i)}^+ = 0$, and so 
\begin{align*}
\sum_{S_{i} \in y'} \kappa_{y \to y'} - \sum_{S_{i} \in y} \kappa_{y \to y'} &= \left[k_{i-1} + k_{i}^- - k_{i}^+\right] + \left[h_{\vp^{-1}(i)-1} + h_{\vp^{-1}(i)}^- - h_{\vp^{-1}(i)}^+\right] =0. 
\end{align*}
\item Finally, for the remaining species $E$, $C_1,\ldots, C_{n-1}$ (including $C_\alpha$) and $D_1, \ldots, D_{m-1}$, balancing of production and consumption rates at the state $\one$ follow immediately from 
\[
k_i^+ = k_i + k_i^- (i \in [1,n-1]), \mbox{ and } h_j^+ = h_j + h_j^- (j \in [1,m-1]).
\]
\enit
For the converse, suppose that $N_\GG$ is not strongly connected. Then there exists a proper subgraph $N'$ of $N_\GG$ such that $N'$ has at least one node, $N'$ is strongly connected, there is at least one edge from a node in $N'$ to a node in $N_\GG \setminus N'$ but there is no edge from a node in $N_\GG \setminus N'$ to a node in $N'$. 
 Since there is no edge from a node in $N_\GG \setminus N'$ to a node in $N'$, the set of nodes in $N'$ must be $\{X_1,\dots,X_\ell\} \cup \{X_{\vp_1},\dots,X_{\vp_p}\}$ for some $1\leq \ell \le n$ and $1\leq p\leq m$, but not both $\ell =n$ and $p=m$. Thus we have
\begin{align*}
\sum_{i: X_i\in N'}& \left(\frac{ds_i}{dt} + \frac{dc_i}{dt}  + \frac{dd_{\vp^{-1}(i)}}{dt} \right)  \\
& = \sum_{i\in[1,\ell]\cap\vp([1,p])}  \left(\frac{ds_i}{dt} +\frac{dc_i}{dt} +\frac{dd_{\vp^{-1}(i)}}{dt}\right) +\sum_{i\in[1,\ell]\cap\vp([1,m])^c}  \left(\frac{ds_i}{dt} +\frac{dc_i}{dt} \right) \\
&+\sum_{i\in[1,n]^c\cap\vp([1,p])}  \left(\frac{ds_i}{dt} +\frac{dd_{\vp^{-1}(i)}}{dt}\right) \\
&= \sum_{i=1}^\ell K_i + \sum_{j=1}^p H_j +\sum_{i:X_i\in N'}\delta_\alpha(i)\sum_{j=1}^mG^\alpha_j.
\end{align*}
At steady state, we must have $\sum_{j=1}^mG^\alpha_j=0$ from Lemma \ref{lem:cascade_ss}. 
Therefore, from Lemma \ref{lem:cascade}, we get that at steady state
\[
0=\sum_{i: X_i\in N'} \left(\frac{ds_i}{dt} + \frac{dc_i}{dt}  + \frac{dd_{\vp^{-1}(i)}}{dt} \right) = \sum_{i=1}^\ell K_i + \sum_{j=1}^p H_j  = -k_\ell c_\ell - h_p d_p,
\]
which implies that $c_\ell$ and $d_p$ are zero at any steady state. 
When $\ell < n$, $c_\ell$ is the concentration of the species $C_\ell$ and when $p < m$, $d_p$ is the concentration of the species $D_p$, and since one of the inequalities must hold, at least one species concentration is zero at steady state. 
In particular, there is no positive steady state, i.e. $\GG$ is not consistent.  
\end{proof}
\begin{remark} \label{rem:construction_existence}
We describe the procedure used to construct rate constants for $\GG$ such that $\one=(1,1,\ldots,1)$ is a steady state. 
If we choose $k_i^- = 1$, $k_i^+ = k_i+1$ for $i \in [1,n-1]$ and $h_j^- = 1$, $h_j^+ = h_j+1$ for $j \in [1,m-1]$, then all enzymes and intermediate complexes are balanced at $\one$. 
Moreover, the net production rate (production rate minus consumption rate) of each $S_i$ is  $-k_{i} + k_{i-1} - h_{\vp^{-1}(i)} + h_{\vp^{-1}(i)-1}$. 
In order to balance the network, i.e. find reaction rate constants such that the net production rate of every species is zero, it suffices to consider the following network, denoted $\XX$, instead:
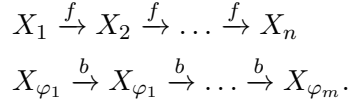
\begin{align*}
&X_1 \xrightarrow{f} X_2 \xrightarrow{f} \ldots \xrightarrow{f} X_n \\
&X_{\vp_1} \xrightarrow{b} X_{\vp_1} \xrightarrow{b} \ldots \xrightarrow{b} X_{\vp_m}. 
\end{align*}
Note that the network $\XX$ above is related to $N_\GG$ appearing in the proof of Theorem \ref{thm:consistency} but is not exactly the same since here we use different `edge types' in the forward and the backward chain. Specifically, each edge is labeled either $f$ or $b$ depending on whether it appears in the forward or the backward chain, to enable a distinction between a transition $X_i \to X_j$ that may appear in both chains. 
Next we construct a cycle by adding a path of edges from $X_n$ to $X_{\vp_1}$ (possibly trivial path if $\vp_1 = n$) and another path of edges from $X_{\vp_m}$ to $X_1$ (possibly trivial path if $\vp_m=1$). 
Such paths exist because $N_\GG$ is strongly connected by hypothesis. 
The added edges may be any edges selected from $\XX$. 
Denote the constructed cycle by $\PP$ where $\PP_{X_i \to X_j}$ is the chosen fixed path from $X_i \to X_j$. 
\[
\PP = X_1 \xrightarrow{f} \ldots \xrightarrow{f} X_n \to (\PP_{X_n \to X_{\vp_1}}) \to  X_{\vp_1} \xrightarrow{b} \ldots \xrightarrow{b} X_{\vp_m} \to (\PP_{X_{\vp_m} \to X_{1}}).
\]
The rate constants of $\GG$ are now determined from the number of times the corresponding edge appears in $\PP$, i.e.
\begin{align*}
k_i &= \abs{\left(X_i \xrightarrow{f} X_{i+1} \right) \in \PP}, \quad \mbox{ and } \\
h_j &= \abs{\left(X_{\vp_j} \xrightarrow{b} X_{\vp_j + 1} \right) \in \PP}.  
\end{align*}
The specific rate constants \eqref{eq:definebackrates} and \eqref{eq:defineforrates} used in the proof of Theorem \ref{thm:consistency} were obtained by constructing specific paths  $\PP_{X_n \to X_{\vp_1}}$ and $\PP_{X_{\vp_m} \to X_{1}}$, as follows:
\begin{align*}
\PP_{X_n \to X_{\vp_1}} = 
\begin{cases}
X_{\vp_{\vp^{-1}(n)}} \xrightarrow{b} X_{\vp_{\vp^{-1}(n)+1}} \xrightarrow{b} \ldots \xrightarrow{b} X_{\vp_{\vp^{-1}(1)}} \xrightarrow{f} X_2 \xrightarrow{f} \ldots \xrightarrow{f} X_{\vp_1} &\mbox{ if } \vp_1 \ne n,\\
\{ \} &\mbox{ if } \vp_1 = n, 
\end{cases} \\
\PP_{X_{\vp_m} \to X_{1}} = 
\begin{cases}
X_{\vp_m} \xrightarrow{f} X_{\vp_{m}+1} \xrightarrow{f} \ldots \xrightarrow{f} X_{n}  \xrightarrow{b} X_{\vp_{\vp^{-1}(n)+1}} \xrightarrow{b} \ldots \xrightarrow{b} X_{\vp_{\vp^{-1}(1)}} &\mbox{ if } \vp_m \ne 1,\\
\{ \} &\mbox{ if } \vp_m = 1. 
\end{cases}
\end{align*}
\end{remark}

\begin{example}
\been
\item Let $\GG$ be the covalent modification network $(5,4,\vp)$ with $\vp([1,4]) = (2,5,1,4)$. $N_\GG$ is clearly strongly connected because it has the edge $X_5 \to X_1$. 

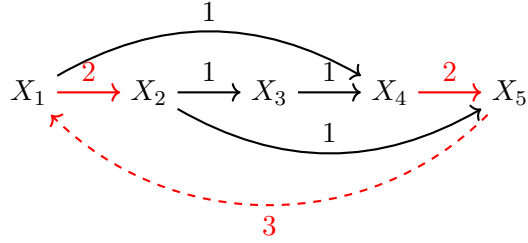
\begin{equation*} 
  \begin{tikzpicture}[baseline={(current bounding box.center)}, scale=0.8]
   \node[state] (X1)  at (0,0)  {$X_1$};
   \node[state] (X2)  at (2,0)  {$X_2$};
   \node[state] (X3)  at (4,0)  {$X_3$};
   \node[state] (X4)  at (6,0)  {$X_4$};
   \node[state] (X5)  at (8,0)  {$X_5$};
   \path[thick,->]
    (X1) edge[->,red] node {$2$} (X2)
        (X2) edge[->] node {$1$} (X3)
            (X3) edge[->] node {$1$} (X4)
                (X4) edge[->,red] node {$2$} (X5)
                    (X2) edge[->,bend right=30] node {$1$} (X5)
                    (X5) edge[->,bend right=-45,red,dashed] node {$3$} (X1)
                    (X1) edge[->,bend right=-30] node {$1$} (X4);
  \end{tikzpicture}
 \end{equation*}
 
\item Let $\GG$ be the covalent modification network $(4,4,\vp)$ with $\vp([1,4]) = (2,1,4,3)$. 
Here the subgraph $N'$ of $N_\GG$ containing nodes $X_1, X_2$ is strongly connected, and there is no edge from a node in $N_\GG\setminus N'$ to a node in $G$. As the proof of Theorem \ref{thm:consistency} suggests, at steady state we have:
\[
0=\frac{ds_1}{dt}+\frac{ds_2}{dt}+\frac{dc_1}{d_t}+\frac{dc_2}{dt} +\frac{dd_1}{dt}+\frac{dd_2}{dt} = -k_2c_2-h_2d_2,
\]
thus $\GG$ does not have any positive steady state, i.e. $\GG$ is not consistent.
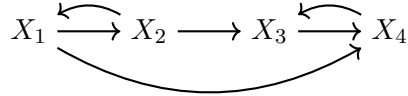
\begin{equation*} 
  \begin{tikzpicture}[baseline={(current bounding box.center)}, scale=0.8]
   \node[state] (X1)  at (0,0)  {$X_1$};
   \node[state] (X2)  at (2,0)  {$X_2$};
   \node[state] (X3)  at (4,0)  {$X_3$};
   \node[state] (X4)  at (6,0)  {$X_4$};
   \path[thick,->]
    (X1) edge[->] node {} (X2)
        (X2) edge[->] node {} (X3)
            (X3) edge[->] node {} (X4)
                    (X2) edge[->,bend right=30] node {} (X1)
                    (X1) edge[->,bend right=30] node {} (X4)
                    (X4) edge[->,bend right=30] node {} (X3);
  \end{tikzpicture}
 \end{equation*}
 
\enen
\end{example}

\begin{corollary}
Let $\GG$ be a consistent covalent modification network $(n,m,\vp)$. Then the $\beta$ and $\gamma$ appearing in Theorem \ref{thm:general} are defined. 
\end{corollary}
\begin{proof}
If $\{1,n\} \not \subseteq \vp([1,m])$ then $N_\GG$ is not strongly connected. 
\end{proof}

\section{Steady state parameterization}\label{sec:parameterization}

We give a steady state parameterization for the futile cycle using only two parameters $e$, the concentration of the enzyme in the forward chain and $u \coloneqq e/c_\alpha$, the ratio between the concentration of the enzyme in the forward chain and backward chain. 
The results of the previous section, especially identity of the ACR species and its ACR value are used in obtaining a steady state parameterization. 
On the other hand, the steady state parameterization is found to be useful for ruling out ACR in other species. 
Additionally, the steady state parameterization helps find the number of positive steady states for futile cycles with a bifunctional enzyme.

Throughout this section, we let $\GG$ be a covalent modification network $(n,m,\vp)$ with a bifunctional enzyme $C_\alpha$, where  $m=n$ and $\vp_j=n+1-j$. The main result is Theorem \ref{thm:existencess}, where we show that a positive steady state exists for all sufficiently large total substrate concentrations. 

\begin{proposition} \label{prop:ssparam}
The steady state concentration of the substrates $S_i$, and the intermediate species $C_i$ and $D_j$ can be expressed in terms of the concentration of the enzymes and their ratio as follows: 
\begin{equation}\label{eq:scdbases2}
\begin{aligned}
s_i &= k_\alpha h^*_{n - \alpha}  \left(\frac{\sigma^{i-1 \downarrow 1}}{\sigma^{\alpha \downarrow 1}} \right) u^{i-\alpha-1}, \quad (i \in [1,n]), \\
c_i &=  \left(\frac{\mu^{i-1 \downarrow 1} }{ \mu^{\alpha-1 \downarrow 1}}\right) e u^{i-\alpha-1}, \quad (i \in [1,n-1]), \\
d_j &=   \frac{k_\alpha}{h_{n-\alpha}} \left( \frac{ \nu^{n-2 \downarrow j}}{ \nu^{n-2 \downarrow n-\alpha}} \right) e u^{n-j-\alpha-1}, \quad (j \in [1,n-1]), 
\end{aligned}
\end{equation}
where for each of $\upsilon \in \{\sigma, \mu, \nu\}$, 
\begin{align}
\upsilon^{i \downarrow j} \coloneqq \begin{cases}  \upsilon_i \cdots \upsilon_j &\mbox{ if } i \ge j, \\
1 &\mbox{ otherwise}, 
\end{cases}
\end{align}
and 
\begin{align}\label{eq:parameters}
\sigma_i = \frac{h_{n-i}^*}{k_i^*}, \quad \mu_i = \frac{k_i h_{n-i}^*}{k_{i+1} k_{i+1}^*}, \quad \nu_i = \frac{h_{i+1} h_{i+1}^*}{h_i k_{n-i}^*}. 
\end{align}
\end{proposition}
\begin{proof}  We observe that for futile cycle, $m=n$ and $\vp([1,m]) = [1,n]$, thus the steady state equations only consist of equation $(1), (4), (5)$ in Corollary \ref{cor:steadystate}. First, we prove that these equations imply \eqref{eq:scdbases2}. We start with a useful claim on recurrence relations  of  the substrates $S_i$, and the intermediate species $C_i$ and $D_j$.

{\bf Claim:} For $i \in [1,n-1]$, we have the following recurrence relations at a positive steady state: 
\begin{align} \label{eq:recrel}
s_{i+1} = u \cdot \sigma_i \cdot s_i, \quad 
c_{i+1} = u \cdot \mu_i \cdot c_i, \quad 
d_{i} = u \cdot \nu_i \cdot d_{i+1}, 
\end{align}
where $\sigma_i,  \mu_i, \nu_i $ are given in \eqref{eq:parameters}. 

{\bf Proof of claim:} 
Let $j=n-i$ in equation $(4)$ in Corollary \ref{cor:steadystate} we obtain $s_{i+1} c_\alpha = h_{n-i}h^*_{n-i} d_{n-i}$. From equation $(5)$ in   Corollary \ref{cor:steadystate} we have $s_ie=k_ik_i^*c_i$. Thus
\begin{equation}\label{eq:ratio_s}
\frac{s_{i+1}}{s_i} = \frac{e}{c_\alpha} \frac{h_{n-i} h^*_{n-i} d_{n-i}}{k_ik^*_ic_i} = u \cdot \sigma_i \cdot \frac{h_{n-i}d_{n-i}}{k_ic_i} =  u \cdot \sigma_i, 
\end{equation}
where the last equality comes from \eqref{eq:base_ratio}. The recurrence relations for $c_i$ come from equation $(5)$ in  Corollary \ref{cor:steadystate} and \eqref{eq:ratio_s}. The recurrence relations for $d_i$ come from equation $(4)$ in  Corollary \ref{cor:steadystate} and \eqref{eq:ratio_s}, which completes proof of the claim.

Next, we return to the proof of the Proposition. The recurrence relations in \eqref{eq:recrel} have the following general solution: 
\begin{equation} \label{eq:scdbases}
\begin{aligned}
s_i = u^{i-1} \sigma^{i-1 \downarrow 1} s_1, \quad 
c_i = u^{i-1} \mu^{i-1 \downarrow 1} c_1, \quad 
d_i = u^{n-i-1} \nu^{n-2 \downarrow i} d_{n-1}. 
\end{aligned}
\end{equation}
From Theorem \ref{thm:base} and \eqref{eq:base_ratio} (where $i=\alpha$) we have the following:
\begin{align}
s_{\alpha + 1} = k_\alpha h^*_{n-\alpha}, \quad 
d_{n-\alpha} = \frac{k_\alpha}{h_{n - \alpha}} c_\alpha. 
\end{align}
Now, we let $i = \alpha+1$ in the first equation, $i= \alpha$ in the second equation and $i = n - \alpha$ in the third equation of  \eqref{eq:scdbases} to get:
\begin{equation}
\begin{aligned}
s_1 &= \frac{s_{\alpha + 1}}{u^\alpha \sigma^{\alpha \downarrow 1}} = \left(\frac{k_\alpha h^*_{n - \alpha}}{\sigma^{\alpha \downarrow 1}} \right) u^{-\alpha}, \\
c_1 &= \frac{c_{\alpha}}{u^{\alpha-1} \mu^{\alpha-1 \downarrow 1}}  = \left(\frac{1}{ \mu^{\alpha-1 \downarrow 1}}\right) e u^{-\alpha}, \\
d_{n-1} &= \frac{d_{n-\alpha}}{u^{\alpha-1} \nu^{n-2 \downarrow n-\alpha}} = \left( \frac{k_\alpha }{h_{n-\alpha}  \nu^{n-2 \downarrow n-\alpha}} \right) e u^{-\alpha}. 
\end{aligned}
\end{equation}
Re-substituting in \eqref{eq:scdbases}, we get the desired equations \eqref{eq:scdbases2}. 

 Conversely, we check that equations \eqref{eq:scdbases2} imply equations $(1),(4),(5)$ in Corollary \ref{cor:steadystate}, and thus they form a parameterization of the steady states. Suppose that $s_i,c_i,d_j$ are given by equations \eqref{eq:scdbases2}. Then we have
\[
k_ic_i = k_i \left(\frac{\mu^{i-1 \downarrow 1} }{ \mu^{\alpha-1 \downarrow 1}}\right) e u^{i-\alpha-1}
\]
and
\[
h_{n-i}d_{n-i} =h_i  \frac{k_\alpha}{h_{n-\alpha}} \left( \frac{ \nu^{n-2 \downarrow n-i}}{ \nu^{n-2 \downarrow n-\alpha}} \right) e u^{i-\alpha-1}
\]
It is straightforward to check that the coefficients are equal:
\[
 k_i \left(\frac{\mu^{i-1 \downarrow 1} }{ \mu^{\alpha-1 \downarrow 1}}\right)= h_i  \frac{k_\alpha}{h_{n-\alpha}} \left( \frac{ \nu^{n-2 \downarrow n-i}}{ \nu^{n-2 \downarrow n-\alpha}} \right), 
\]
thus $k_ic_i=h_{n-i}d_{n-i}$ for all $i$. This further implies
\[k_{i-1}c_{i-1}=h_{n-i+1}d_{n-i+1}-h_{n-i}d_{n-i} = -h_{\vp^{-1}(i)-1}d_{\vp^{-1}(i)-1}+h_{\vp^{-1}(i)}d_{\vp^{-1}(i)},\]
which is equation $(1)$ in Corollary \ref{cor:steadystate}.

Next, we have
\[
s_{i+1}c_\alpha = k_\alpha h^*_{n-\alpha} \left(\frac{\sigma^{i \downarrow 1}}{\sigma^{\alpha \downarrow 1}} \right) u^{i-\alpha}c_\alpha =  k_\alpha h^*_{n-\alpha} \left(\frac{\sigma^{i \downarrow 1}}{\sigma^{\alpha \downarrow 1}} \right) e u^{i-\alpha-1},
\]
where the last equality comes from $c_\alpha=e/u$. We also have
\[
h_{n-i}h^*_{n-i} d_{n-i}= h_{n-i}h^*_{n-i} \frac{k_\alpha}{h_{n-\alpha}} \left( \frac{ \nu^{n-2 \downarrow n-i}}{ \nu^{n-2 \downarrow n-\alpha}} \right) e u^{i-\alpha-1}.
\]
Again, it is straightforward to verify that the coefficients in $s_{i+1}c_\alpha$ and $h_{n-i}h^*_{n-i} d_{n-i}$ are the same, thus $s_{i+1}c_\alpha=h_{n-i}h^*_{n-i} d_{n-i}$. By substituting $j=n-i$ we obtain equation $(4)$ in Corollary \ref{cor:steadystate}. Similarly, we can check that $s_ie = k_ik_i^*c_i$ by verifying that they have the same monomial in $e$ and $u$ and same coefficient.

\end{proof}

\begin{proposition} \label{prop:polyg}
Within a given compatibility class, with fixed total substrate $T_s$ and total enzyme $T_e$, defined by the equations
\begin{align} \label{eq:consquant}
\sum_{i=1}^ns_i+\sum_{i=1}^{n-1}c_i+2\sum_{i=1}^{i-1}d_i = T_s, \quad \quad \quad 
e+ \sum_{i=1}^{n-1}c_i+\sum_{i=1}^{i-1}d_i = T_e, 
\end{align}
the number of positive steady states is the number of positive solutions of the following equation in the variable $u$
\begin{align} \label{eq:polygu}
&g_\alpha(u) \coloneqq \sum_{i=0}^{n-1} \left(\frac{T_s \delta_\alpha - a_i}{T_e} \right) u^i = \frac{u^{\alpha} \sum_{i=0}^{n-2} b_i u^i}{u^{\alpha} +  \sum_{i=0}^{n-2} q_i u^i } \eqqcolon h_\alpha(u)
\end{align}
where 
\begin{align*}
a_i =  k_\alpha h^*_{n - \alpha}  \left(\frac{\sigma^{i \downarrow 1}}{\sigma^{\alpha \downarrow 1}} \right), \quad
b_i = \frac{\mu^{i \downarrow 1}}{\mu^{\alpha-1 \downarrow 1}} + 2\left( \frac{k_\alpha }{h_{n-\alpha} } \right) \frac{\nu^{n-2 \downarrow n-i-1}}{\nu^{n-2 \downarrow n-\alpha}}, \quad
q_i = \frac{\mu^{i \downarrow 1}}{\mu^{\alpha-1 \downarrow 1}} + \left( \frac{k_\alpha }{h_{n-\alpha} } \right) \frac{\nu^{n-2 \downarrow n-i-1}}{\nu^{n-2 \downarrow n-\alpha}}. 
\end{align*}
\end{proposition}
\begin{proof}
Plugging the solution \eqref{eq:scdbases2} into the conservation laws \eqref{eq:consquant} results in 
\begin{equation}
\begin{aligned}  \label{eq:cons4}
 \sum_{i=0}^{n-1}a_i u^i + e\sum_{i=0}^{n-2} b_i u^i  &= T_su^\alpha, \quad 
e \left[1 + u^{-\alpha} \sum_{i=0}^{n-2} q_i u^i  \right] &= T_e, 
\end{aligned}
\end{equation}
where $a_i, b_i$ and $q_i$ are as defined in the statement of the theorem. Eliminating the variable $e$, and multiplying through by the common denominator gives an equation whose left side is \eqref{eq:polygu} and the right side is 0. For every positive zero $u^*$ of \eqref{eq:polygu}, we can find the value of the coordinate $e$ from the second equation in \eqref{eq:cons4} and finally the other coordinates from the steady state parameterization \eqref{eq:scdbases2} appearing in Proposition \ref{prop:ssparam}. 
It is clear that these other coordinates are positive when $u^*>0$.
\end{proof}
\begin{theorem} \label{thm:existencess}
Let $\GG$ be a covalent modification network $(n,m,\vp)$ with a bifunctional enzyme $C_\alpha$, where  $m=n$ and $\vp_j=n+1-j$. 
\been
\item Suppose that $\alpha = n-1$. Then a positive steady state exists if and only if $T_s > k_\alpha h^*_{n - \alpha}$. Moreover, when a positive steady state exists, the number (counted with multiplicity) of positive steady states is odd. 
\item Suppose that $\alpha \in [1,n-2]$. Then for every fixed $T_e$, there exists a $\widehat T_s > k_\alpha h^*_{n - \alpha}$ such that a positive steady state exists if $T_s > \widehat T_s$. 
Moreover, when a positive steady state exists, the number (counted with multiplicity) of positive steady states is even. 
\enen
\end{theorem}
\begin{proof}
When $\alpha = n-1$, \eqref{eq:polygu} becomes 
\begin{align} \label{eq:polygun-1}
&g_{n-1}(u) = \frac{T_s  - a_{n-1}}{T_e} u^{n-1} - \sum_{i=0}^{n-2} \left(\frac{a_i}{T_e} \right) u^i = \frac{u^{n-1} \sum_{i=0}^{n-2} b_i u^i}{u^{n-1} +  \sum_{i=0}^{n-2} q_i u^i } = h_{n-1}(u).
\end{align}
For $u > 0$, the range of $h_{n-1}$ is $(0,\infty)$. 
If $T_s \le a_{n-1}$, then $g_{n-1}(u) < 0$ for all $u > 0$ and so there is no positive solution. 
Now suppose that $T_s > a_{n-1}$. 
Then $g_{n-1}(u) \sim c' u^{n-1}$ (as $u \to \infty$) for some positive constant $c'$.
Note that $g_{n-1}(0) < 0 = h_{n-1}(0)$. 
Since $h_{n-1}(u) \sim c''u^{n-2}$ and $g_{n-1}(u) \sim c'u^{n-1}$, $g_{n-1}(u) > h_{n-1} (u)$ for all sufficiently large $u$. 
It follows \eqref{eq:polygun-1} must have an odd number (counted with multiplicity) of positive solutions, and therefore at least one.  
By Proposition \ref{prop:polyg}, the number of positive solutions is the same as the number of positive steady states.
Finally, note that $a_{\alpha} = a_{n-1} = k_\alpha h^*_{n - \alpha}$ is both the ACR value and the threshold for existence of positive steady state.

For any $\alpha \in [1,n-2]$, $g_\alpha (0) = - a_0/T_e < 0 = h_\alpha(0)$. 
Moreover, $g_\alpha(u) \sim - (a_{n-1}/T_e) u^{n-1}$ while $h_\alpha(u)$ is positive on $(0,\infty)$.
Therefore, the number (counted with multiplicity) of positive solutions of \eqref{eq:polygu} must be even. 

To show that a positive solution exists for a large enough $T_s$, we only need to show that the graph of $g_\alpha$ crosses the graph of $h_\alpha$. Indeed, fix $T_e > 0$ and let 
\[
\widehat T_s = \sum_{i=0}^{n-1} a_i + \frac{T_e \sum_{i=0}^{n-2} b_i}{1 + \sum_{i=0}^{n-2} q_i}. 
\]
Then for any $T_s > \widehat T_s$, we have 
\begin{align*}
T_e \left(g_\alpha(1) - h_\alpha(1)\right) = T_s - \sum_{i=0}^{n-1} a_i - \frac{T_e \sum_{i=0}^{n-2} b_i}{1 + \sum_{i=0}^{n-2} q_i} = T_s - \widehat T_s > 0.
\end{align*}
It follows that $g_\alpha(1) > h_\alpha(1)$ and therefore there is a positive solution for $u \in (0,1)$ and another positive solution in the interval $(1,\infty)$. 
\end{proof}

\section{Existence of boundary steady state}\label{sec:boundaryss}

\begin{definition}
Let $\GG$ be a covalent modification network $(n,m,\vp)$. We denote by $x_B$ the concentration vector where $s_i=0$ for any $i\neq n$; $c_i=0$ for  $i\in[1,n-1]$, $d_j=0$ for $j\in[1,m-1]$, $s_n=T_s$ and $e=T_e$.
\end{definition}

\begin{lemma}
Let $\GG$ be a covalent modification network $(n,m,\vp)$. Then the concentration vector $x_B$ is a boundary steady state of $(\GG,(k,h))$. 
\end{lemma}
\begin{proof}
It is easy to check that $x_B$ satisfies all five equations in Corollary \ref{cor:steadystate}.
\end{proof}

\begin{lemma}\label{lem:e>0}
Let $\GG$ be a covalent modification network $(n,m,\vp)$. Suppose that $T_e>0$.  Then at any boundary steady state, we must have $e>0$. 
\end{lemma}
\begin{proof}
Assume by contradiction that $e=0$ at a boundary steady state. From equation $(5)$ in Corollary \ref{cor:steadystate}, we have $c_i=0$ for $i\in [1,n-1]$. In particular, $c_\alpha=0$ and from equation $(4)$ in Corollary \ref{cor:steadystate} we have $d_j=0$ for $j\in[1,m-1]$. Thus $T_e=e+\sum_{i=1}^{n-1}c_i+\sum_{j=1}^{m-1}d_j=0$, which contradicts the assumption $T_e>0$.
\end{proof}

\begin{lemma}\label{lem:leftside}
Let $\GG$ be a covalent modification network $(n,m,\vp)$. Suppose that $T_e>0$. Then at any boundary steady state, the following statements hold
\begin{enumerate}
\item For any $\ell\in[1,n-1]$, if $c_\ell=0$ or $s_\ell=0$, then $c_i=s_i=0$ for $i\leq \ell$.
\item For any $\ell\in[1,m-1]$, if $c_\alpha \neq 0$ and $d_\ell=0$ or $s_{\vp_\ell}=0$, then $d_j=0$ and $s_{\vp_j}=0$ for $j\leq \ell$.
\end{enumerate}
\end{lemma}

\begin{proof}
From Lemma \ref{lem:e>0}, we must have $e>0$. For part (1), due to equation $(5)$ in Cor \ref{cor:steadystate}, $c_\ell=0$ if and only if $s_\ell=0$. Assume that $c_\ell=s_\ell=0$. Since either $F_\ell=0$ if $\ell\in \vp([1,m])^c$ or $F_\ell=-G^\alpha_{\vp^{-1}(\ell)} = h^+_{\vp^{-1}(\ell)}s_\ell c_\alpha - h^-_{\vp^{-1}(\ell)} d_{\vp^{-1}(\ell)} - h^-_{\vp^{-1}(\ell)-1} d_{\vp^{-1}(\ell)-1} = - h^-_{\vp^{-1}(\ell)} d_{\vp^{-1}(\ell)} - h^-_{\vp^{-1}(\ell)-1} d_{\vp^{-1}(\ell)-1}$, we must have $F_\ell \leq 0$. On the other hand, we have $F_\ell = -k_\ell^+s_\ell e + k^-_\ell c_\ell + k_{\ell-1}c_{\ell -1} = k_{\ell-1}c_{\ell -1} \geq 0$. Thus $F_\ell=0$ and consequently $c_{\ell-1}=0$. By a simple induction argument, we have $c_i=s_i=0$ for $i\leq \ell$.

The proof for part (2) is similar and thus is omitted for the sake of brevity.
\end{proof}

\begin{theorem}
Let $\GG$ be a covalent modification network $(n,m,\vp)$ with auxiliary graph $N_\GG$. Suppose that $N_\GG$ is strongly connected and $T_s>0, T_e>0$. Then $x_B$ is the only boundary steady state $\GG$ can admit.
\end{theorem}
\begin{proof}
Consider a boundary steady state that $\GG$ can admit. From  Lemma \ref{lem:e>0}, we must have $e>0$. We consider two cases below.

\noindent \textbf{Case 1:} Suppose that $c_\alpha=0$. Then from equation $(4)$ in Corollary \ref{cor:steadystate}, we must have $d_j=0$ for $j\in[1,m-1]$. From equations $(1)$ and $(2)$ in Corollary \ref{cor:steadystate}, this further implies $k_{i-1}c_{i-1}=k_ic_i$ for $i\in[1,n-1]$. Thus $c_i=0$ for $i\in[1,n-1]$ and by Lemma \ref{lem:leftside} we have $s_i=0$ for $i\in[1,n-1]$. As a result, $s_n=T_s$ and $e=T_e$ and the boundary steady state must be $x_B$.

\vspace{.1in}
\noindent \textbf{Case 2:} Suppose that $c_\alpha> 0$. From equations $(4)$ and $(5)$ in Corollary \ref{cor:steadystate}, it suffices to assume that $s_\ell=0$ for some $\ell \in [1,n]\cap \vp[1,m]^c$. Since the auxiliary graph $N_\GG$ is strongly connected, there exists a path from $X_n$ to $X_\ell$. Together with Lemma \ref{lem:leftside}, this implies $s_n=0$ and thus $c_\alpha=0$. We reach a contradiction in this case.
\end{proof}

\section{Discussion}

In this paper, we focus on bifunctional enzyme action, an important mechanism which has been shown to cause robustness in  biological networks. Intuitively, a bifunctional enzyme in native form facilitates the production of a substrate, while in bound form catalyzes the degradation of the same substrate. We have shown that bifunctional enzyme action can ensure absolute concentration robustness (ACR) of a target substrate in a large class of covalent modification networks. Our main results not only state sufficient conditions for the existence of ACR in this class of networks, but also pinpoint precisely the ACR species and provide the ACR value (i.e. the steady state concentration of the robust species). Notably, our results do not rely on the deficiency of the networks like the well-known Shinar-Feinberg criterion in \cite{shinar2010structural}, and thus they can be applied to biological networks with various size and complexity.

In addition, we have provided the necessary and sufficient conditions for the existence of a positive steady state, and the existence and uniqueness of a boundary steady state, in covalent modification networks. For a special subclass consisting of futile cycles with a bifunctional enzyme, we have also given a steady state parameterization based on the bifunctional enzyme concentration and the ratio between the native form and the bound form of this enzyme. This steady state parameterization indicates that futile cycles with bifunctional enzyme can be multistationary (i.e. have multiple positive steady states).

Going forward, we are planning to extend the definition of bifunctionality and our results on ACR to a significantly more general class of enzymatic networks. For example, we can allow for different enzymes (see Remark \ref{rem:multi_enzyme}) and/or multiple intermediate complexes in each step of the reaction cascade. 
The intermediate steps of an enzyme-catalyzed reaction often vary between different modeling choices since they are difficult to pin down experimentally. 
Our goal is to develop results on the connection between bifunctionality and ACR in systems which have minimal underlying assumptions on such intermediate catalysis steps. 


Another direction we plan to pursue involves studying the dynamics of networks with bifunctional enzyme. Can they exhibit important dynamical properties like bistability or dynamic ACR (for example, see \cite{joshi2022foundations,joshi2023motifs,joshi2023power})? What is the stability of the boundary steady state? These questions would shed light on how the presence of bifunctionality impact the dynamics of biological networks besides causing ACR.

\section*{Statements and Declarations}

\subsection*{Competing interests}
The authors have no relevant financial or non-financial interests to disclose. 
\subsection*{Data availability}
The authors declare that the data supporting the findings of this study are available within the article.

\section*{Acknowledgments}
B.J. was supported by NSF grant DMS-2051498. 

\bibliographystyle{unsrt}
\bibliography{acr}

\end{document}